\documentclass[10pt]{article}
\usepackage{amssymb,amsmath,amsthm}
\usepackage[english]{babel}
\usepackage{booktabs}
\usepackage{caption}
\usepackage{enumitem}
\usepackage{geometry}
\usepackage{graphicx}
\usepackage[colorlinks,citecolor=blue,urlcolor=blue]{hyperref}
\usepackage[utf8]{inputenc}
\usepackage{makecell}
\usepackage{mathtools}
\usepackage[round]{natbib}
\usepackage[labelformat=simple]{subcaption}
\usepackage{url}

\usepackage{cleveref}

\usepackage{stevemath}

\newcommand{\unif}{{\textup{unif}}}
\newcommand{\Gammatw}{\widetilde{\Gamma}}

\begin{document}

\title{The uniform general signed rank test and its design sensitivity}
\author{Steven R. Howard \and Samuel D. Pimentel}
\maketitle

\begin{abstract}
A sensitivity analysis in an observational study tests whether the qualitative
  conclusions of an analysis would change if we were to allow for the
  possibility of limited bias due to confounding. The design sensitivity of a
  hypothesis test quantifies the asymptotic performance of the test in a
  sensitivity analysis against a particular alternative. We propose a new,
  nonasymptotic, distribution-free test, the uniform general signed rank test,
  for observational studies with paired data, and examine its performance under
  Rosenbaum's sensitivity analysis model. Our test can be viewed as adaptively
  choosing from among a large underlying family of signed rank tests, and we
  show that the uniform test achieves design sensitivity equal to the maximum
  design sensitivity over the underlying family of signed rank tests. Our test
  thus achieves superior design sensitivity, indicating
  it will perform well in sensitivity analyses on large samples. We support this
  conclusion with simulations and a data example, showing that the advantages of
  our test extend to moderate sample sizes as well.
\end{abstract}

\section{Introduction}
In the empirical study of causal effects, the use of standard statistical
hypothesis tests, along with their concomitant $p$-values and confidence
intervals, accounts only for the uncertainty introduced by sampling
variability. However, in an observational study where treatment assignment has
not been randomized, hidden biases due to unobserved confounding can be much
larger than sampling uncertainty. As such, standard hypothesis tests may fail to
be convincing if they assume the study is free of hidden bias, as a randomized
experiment would be. A sensitivity analysis addresses this problem by formally
testing whether the qualitative conclusions of a standard procedure would change
if hidden bias of a certain magnitude were present
\citep{rosenbaum_observational_2002}.

When an investigator plans to run a sensitivity analysis, the choice of test
statistic may no longer hinge solely on traditional measures such as Pitman
efficiency \citep{pitman1948lecture, nikitin_efficiency_2011}. In particular, an investigator may seek a test statistic that is
least sensitive to hidden bias, and thereby most likely to successfully
distinguish treatment effects from bias, rather than one that is most likely to
detect treatment effects in the absence of hidden bias. Design sensitivity is
one way to quantify this idea for a particular test statistic
\citep{rosenbaum_design_2004, rosenbaum_design_2010}. Design sensitivity
complements Pitman efficiency and other conventional means of comparing tests.

\citet{rosenbaum_design_2010-1} shows that a test statistic that focuses on a
subgroup strongly affected by treatment may achieve superior design sensitivity
compared to a statistic that uses all observations. \citet{rosenbaum_exact_2012}
shows that a particular test using only the observations of largest magnitude,
Noether's test \citep{noether_simple_1973}, has excellent design sensitivity but
poor power against small effects. Rosenbaum then proposes an adaptive test in
which the $p$-value is the minimum $p$-value from two competing test statistics,
corrected for multiple testing using the joint distribution of these two test
statistics. This adaptive test is shown to get the best of both worlds, good
power in small samples as well as high design sensitivity. In fact, the adaptive
test attains the maximum design sensitivity of its two component
tests. \citet{rosenbaum_adaptive_2017} similarly propose an adaptive test that
chooses from the better of two test statistics, one focused on a subgroup and
one examining the entire population, with correction for multiple testing.

We examine a different test for paired data that chooses
adaptively from a large, highly dependent family of test
statistics. We control for multiple testing using a uniform concentration bound
for the stochastic process formed by this family of test statistics. This
permits choosing among as many test statistics as we have
observations, while achieving nonasymptotic, distribution-free error
control. Our theoretical results characterize how this test achieves excellent
design sensitivity, 
which can be infinite against normal
alternative distributions, such that no matter the strength of confounding, the test will
reject the null hypothesis of no effect with probability approaching one asymptotically. We are not aware of
previous discussion of such behavior.

\section{Background and notation}\label{sec:background}

\subsection{Rosenbaum's sensitivity analysis model for matched pairs}

We focus attention on observational studies in which study subjects receiving a
 treatment condition are paired to similar subjects receiving a control condition for
 analysis.
 For instance, in the study of fish consumption and mercury concentration in \Cref{sec:data}, respondents in a  
 nutrition survey who consume 15 servings of fish per month are each paired to a
 respondent who consumes two or fewer servings per month, but has similar
 demographic attributes and smoking habits. 
We 
assume $n$ such pairs have been constructed. The
subjects in the $i$th pair have control potential outcomes $R_{Cij}$, treatment
potential outcomes $R_{Tij}$, and treatment indicators $Z_{ij}$ for $j = 1, 2$
and $i \in [n]$ where $[n]$ is the set $\{1, \ldots, n\}$ and where $Z_{i1} + Z_{i2} = 1$ for all pairs $i$ by construction. Let $\Fcal$ be the $\sigma$-field generated by all of the
potential outcomes $(R_{Cij}, R_{Tij})_{i \in [n], j \in [2]}$.

A sensitivity analysis tests whether a positive conclusion of our
study, specifically a rejection of Fisher's sharp null hypothesis of no effect of treatment for any individual, holds up under the possibility of
limited confounding \citep{rosenbaum_observational_2002}. To operationalize this notion, for each $\Gamma \geq 1$ we
define the sensitivity analysis null hypothesis $H_0(\Gamma)$ to assert, firstly, that
that
$R_{Ti1} = R_{Ci1}$ and $R_{Ti2} = R_{Ci2}$ for all $i \in [n]$, which is
  Fisher's sharp null, and, secondly, that
conditional on $\Fcal$, treatment assignments are independent between
  pairs, and that treatment probabilities within each pair $i$ are related by the
  following odds ratio bounds:
  \begin{align*}
    \frac{1}{\Gamma} \leq
    \frac{\P\condparen{Z_{i1} = 1}{\Fcal} / \P\condparen{Z_{i1} = 0}{\Fcal}}
         {\P\condparen{Z_{i2} = 1}{\Fcal} / \P\condparen{Z_{i2} = 0}{\Fcal}}
    \leq \Gamma.
  \end{align*}
At $\Gamma = 1$, this specifies that, within each pair, both units have the same
(conditional) probability of treatment. This is equivalent to assuming that treatments are
assigned completely at random within pairs, 
the standard assumption that
leads to valid randomization inference in the absence of unmeasured confounding variables
\citep[\S 3.2]{rosenbaum_observational_2002}.  When $\Gamma > 1$, treatment probabilities may differ within a pair in ways
we cannot observe. If differences in treatment probabilities within pairs are correlated with differences in potential outcomes,  bias will arise
in estimates of the effect of treatment, although the magnitude of
such bias is limited by the sensitivity parameter $\Gamma$.  An equivalent model, described in \citet[\S 4.2]{rosenbaum_observational_2002}, assumes the presence of an unmeasured confounding variable whose relationship to treatment is limited by parameter $\Gamma$, and whose relationship to the outcome variable is unrestricted, allowing it to be arbitrarily strong.  Additional assumptions about the strength of the confounder-outcome relationship can be incorporated by an amplification of the sensitivity analysis as in \citet{rosenbaum2009amplification}.  

Write $R^\obs_{ij} = Z_{ij} R_{Tij} + (1 - Z_{ij}) R_{Cij}$ for the
observed outcomes and $Y_i = (Z_{i1} - Z_{i2})(R^\obs_{i1} - R^\obs_{i2})$ for
the observed treated-minus-control difference in the $i$th pair.  Under the null
$H_0(\Gamma)$ we know that $Y_i \in \left\{\abs{R_{Ci1} - R_{Ci2}}, -\abs{R_{Ci1} - R_{Ci2}}\right\}$ and
\begin{align*}
\frac{1}{1 + \Gamma} \leq
\text{pr}\condparen{Y_i > 0}{\Fcal}
\leq \frac{\Gamma}{1 + \Gamma}.
\end{align*}
Here and in all future statements, we implicitly condition on the event
$\brace{Z_{i1} + Z_{i2} = 1 \text{ for all }i \in [n]}$, omitting it from the notation, and we also assume $\text{pr}(Y_i = 0) = 0$ for all $i$.  
 In words, $H_0(\Gamma)$ asserts that there is no effect of treatment for
any individual, but that treatment probabilities may differ within a pair in ways we cannot observe.
   Again, $\Gamma = 1$
rules out unobserved confounding,
since $\text{pr}\condparen{Y_i > 0}{\Fcal, Z_{i1} + Z_{i2} = 1} =
1/2$; if $\Gamma > 1$ and potential outcomes are correlated with unobserved differences in treatment probabilities, bias due to confounding will be present. 

This sensitivity analysis model provides a method to conduct hypothesis
tests that control type I error under limited confounding, but leaves open the choice of test
statistic.To judge the relative benefits of different test statistics,
we calculate power for various test statistics
in a test of the sensitivity analysis null $H_0(\Gamma)$. 
As with all power
calculations, we must choose a particular alternative hypothesis of interest. We define an alternative hypothesis $H_1(G)$ for a
distribution $G$ over $\R$, motivated by the following scenario, in which we hope to reject the null hypothesis with high power.  Firstly, we let $R_{Cij}$ be an independent draw from some distribution $F$, for each
  $i \in [n], j = 1, 2$.  Secondly, we specify $R_{Tij} = R_{Cij} + \tau_i$ for all $i,j$, where $\tau_i \in \R$ is drawn
  from some fixed distribution for each $i \in [n]$, and is constant within each
  pair.  Finally, we require
  $\text{pr}\condparen{Z_{i1} = 1, Z_{i2} = 0}{\Fcal} = \text{pr}\condparen{Z_{i1} = 0, Z_{i2}
    = 1}{\Fcal} = 1/2$, with treatment (conditionally) independent between
  pairs.
In words,  there is a constant treatment effect within pairs and no hidden bias
due to unequal treatment probabilities.  The alternative hypothesis $H_1(G)$ is
characterized by the induced distribution $G$ of the independent and identically distributed pair
differences $Y_i = (Z_{i1} - Z_{i2})(R_{Ci1} - R_{Ci2}) + \tau_i$. Because there
is no hidden bias under $H_1(G)$, the mean of this distribution, when the mean exists, is the
average treatment effect $E\left(\tau_i\right)$. We explore the performance of the test when
$\tau_i =  \tau$  is constant across pairs, so that $G$ is the distribution of
$R - R' + \tau$, where $R$ and $R'$ are independent draws from $F$; this
distribution is symmetric about $\tau$. We also explore performance under a rare effects
model in which $\tau_i$ is zero for most pairs and equal to some large value for
a small proportion of pairs. In this case, $G$ is a mixture with most mass
placed on some distribution symmetric about zero, and the remaining mass on a
copy of the distribution shifted to the right.

When $\Gamma$ is sufficiently small, power tends to approach one as $n$ grows large under either alternative described above, but for most test statistics there is some threshold value of $\Gamma$ above which power instead tends to zero.  This threshold is known as the design sensitivity and depends not only on the alternative but on the test statistic \citep[\S 14]{rosenbaum_design_2010}.  Test statistics with large design sensitivities are desirable because they report the presence of a treatment effect with high probability under the alternative, even allowing for large amounts of unmeasured confounding.

Rosenbaum's sensitivity analysis model is only one of many possible
approaches.   \citet{fogarty2019studentized} discusses the performance of the Rosenbaum procedure for a weak null hypothesis rather than a sharp null. \citet{cornfield_smoking_1959} and \citet{ding_sensitivity_2016} develop sensitivity bounds on a relative risk using sensitivity parameters defined by the relative risk of exposure, or outcome, in categories of an unobserved confounding variable, \citet{gilbert_sensitivity_2003} give an approach based on logistic modeling, and \citet{robins_sensitivity_2000} and \citet{yu_sensitivity_2005} consider methods appropriate for time-varying treatments and doses of treatment respectively.

\subsection{Sensitivity analysis with general signed rank statistics}
\label{sec:general_sign_rank}

\begin{figure}
  \includegraphics{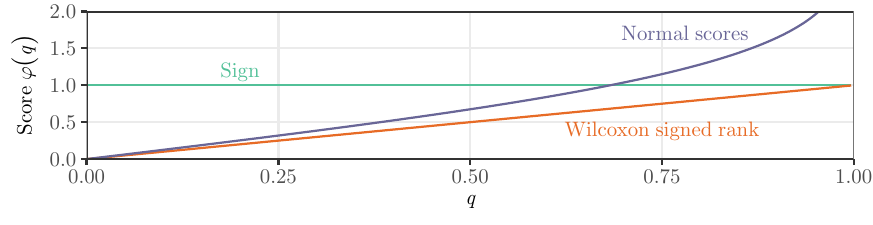}
  \caption{The three score functions $\varphi(q)$ used in this
    paper. \label{fig:score_fns}}
\end{figure}

Let $(Y_{(i)})$ denote the pair differences $(Y_i)$ ordered by absolute value,
so that $\abs{Y_{(1)}} \leq \dots \leq \abs{Y_{(n)}}$. A
{general signed rank statistic} has the form
\begin{align*}
  T_n = \sum_{i=1}^n \varphi\pfrac{i}{n + 1} \indicator{Y_{(i)} > 0}
\end{align*}
for some {score function} $\varphi: (0,1) \to [0,\infty)$. See \citet[\S
6.10]{lehmann_testing_2005} and references therein for pointers to the long
history of general signed rank tests; \citet{rosenbaum_design_2010-1} discusses
their use in the context of sensitivity analysis. The score function $\varphi$
allows us to place more or less weight on pairs with relatively larger or
smaller observed absolute differences. We will focus on the three score
functions underlying the sign test, the Wilcoxon signed rank test, and the
normal scores test, although we discuss an additional, redescending score
function in \cref{sec:redescending}.  The {sign
  test} uses $\varphi(q) =1$, so that all pairs contribute equally, regardless
of rank. In this case $T_n$ simply counts the number of pairs in which the
treated unit had a higher outcome.  The {Wilcoxon signed rank test} is
equivalent to $\varphi(q) = q$ \citep{rosenbaum_design_2010-1}, so that pairs
with larger effects contribute more to the test statistic.  Finally, the {normal
  scores} test uses $\varphi(q) = \Phi^{-1}\lbrace(1 + q) / 2\rbrace$, where
$\Phi^{-1}$ is the standard normal quantile function,
$\text{pr}\lbrace Z \leq \Phi^{-1}(q)\rbrace = q$ when $Z \sim
\Normal(0,1)$. This score function is the quantile function of the absolute
value of a standard normal random variable, and this general signed rank test
has high power when outcomes are drawn from a normal distribution \citep[\S
6.9-6.10]{lehmann_testing_2005}.  All three score functions are illustrated in
\cref{fig:score_fns}.

The sensitivity analysis null hypothesis $H_0(\Gamma)$ does not specify a single
distribution for the observables $(Y_i)$, but it does allow for easy construction of a single worst-case
distribution for the test statistic $T_n$ in a one-sided test that
rejects for $T_n$ sufficiently large, that is, a distribution that
maximizes $\text{pr}\condparen{T_n \geq a}{\Fcal}$ for any threshold $a$, among
all distributions in $H_0(\Gamma)$. This worst-case distribution has the $n$
signs $(\indicator{Y_i > 0})$ independent with
$\text{pr}\condparen{Y_i > 0}{\Fcal} = \Gamma / (1 + \Gamma)$ for all $i \in [n]$
\citep[\S 4.3]{rosenbaum_observational_2002}. Write $c_{\alpha,n}(\Gamma)$
for the $1-\alpha$ quantile of $T_n$ under this worst-case distribution,
so that $c_{\alpha,n}(\Gamma)$ is the critical value of a one-sided,
level-$\alpha$ sensitivity analysis testing $H_0(\Gamma)$ with test statistic
$T_n$; the critical value may depend on $\Fcal$, in the case of ties. This
critical value yields a valid test of the sensitivity analysis
null hypothesis, and is not hard to approximate numerically or via the normal
distribution. In \cref{th:uniform_bound} below, we build upon these ideas to
define a uniform general signed rank test, deriving closed-form critical values
that guarantee nonasymptotic Type I error control under the sensitivity null
$H_0(\Gamma)$.

\subsection{Power of a sensitivity analysis and design sensitivity}
\label{sec:power}

The power of a one-sided, level-$\alpha$ sensitivity analysis
for a general signed rank test with statistic $T_n$ is
$\text{pr}_1\brace{T_n \geq c_{\alpha,n}(\Gamma)}$, where $\text{pr}_1(A)$ is the probability of an event $A$ under
$H_1(G)$.  The power is well-defined, since
$H_1(G)$ specifies the distribution of $T_n$ completely, and
depends on the level $\alpha$, the sample size $n$, the sensitivity parameter
$\Gamma$, the alternative distribution $G$, and the score function
$\varphi$. The {design sensitivity} \citep{rosenbaum_design_2004,
  rosenbaum_design_2010} of the test statistic $T_n$ is the value
$\Gammatw$ such that, as the sample size grows without bound, the power
of a sensitivity analysis with parameter $\Gamma$ approaches one whenever
$\Gamma < \Gammatw$ and approaches zero whenever
$\Gamma > \Gammatw$:
\begin{align*}
\lim_{n \to \infty} \text{pr}_1\ebrace{T_n \geq c_{\alpha,n}(\Gamma)} &= \left\{ \begin{array}{rl} 1, & 
  \quad 1 \leq \Gamma < \Gammatw, \\
0,&
  \quad \Gammatw < \Gamma < \infty. \end{array} \right.
\end{align*}
Formally, the design sensitivity depends on the level $\alpha$, the alternative
distribution $G$ and the score function $\varphi$. In typical examples,
including those considered below, the dependence on $\alpha$ vanishes; intuitively, the design sensitivity tends to 
delineate the point at which the worst-case mean of the test statistic under the sensitivity null $H_0(\Gamma)$ is equal to the mean under the alternative, and this is invariant to $\alpha$. It is
clear from the definition that such a value is unique, if it exists, but
existence must be proved as part of the derivation of design sensitivity, as in
our \cref{th:unif_ds}. Note also that we may have $\Gammatw = \infty$,
which means that
$\lim_{n \to \infty} \text{pr}_1\ebrace{T_n \geq c_{\alpha,n}(\Gamma)} = 1$ for all
$\Gamma \geq 1$; in words, the test has power approaching one against the given
alternative regardless of how large a sensitivity parameter $\Gamma$ is chosen.

Proposition 2 of \citet{rosenbaum_design_2010-1} gives a formula for the design
sensitivity of a general signed rank test whenever the score function $\varphi$
is piecewise continuous, nondecreasing and not identically zero:
\begin{align}
\Gammatw = \frac{\pi}{1 - \pi}, \quad 
\quad
\pi =
  \frac{\int_0^\infty \varphi\ebrace{G(y) - G(-y)} \d G(y)}{\int_0^1 \varphi(y) \d y}.
\label{eq:fixed_design_sensitivity}
\end{align}
Note that $G(y) - G(-y)$ is the CDF of $\abs{Y}$ under $H_1(G)$. We see that the
design sensitivity of a general signed rank test is determined precisely by the
aspects of $\varphi$ and $G$ captured in the quantity $\pi$. In
\cref{th:unif_ds} and \cref{th:limiting_ds}, we extend this result to
characterize the design sensitivity of our uniform general signed rank test. Our
conditions on $\varphi$, while not strictly more general, do allow for the
normal score function, in contrast to Rosenbaum's
conditions.

For the sign test, $\varphi(q) \equiv 1$, we have $\int_0^1 \varphi(y) \d y = 1$
and $\int_0^\infty \varphi\ebrace{G(y) - G(-y)} \d G(y)$ is exactly
$\text{pr}(Y > 0)$ when $Y \sim G$. Hence $\pi = \text{pr}_1(Y > 0)$
(\citealp{rosenbaum_exact_2012}, Proposition 1). Both $\pi$ and
$\widetilde{\Gamma}$ express the chance that a pair difference $Y$ gives
evidence in favor of a positive treatment effect under the alternative with no
hidden bias, as a probability and as an odds respectively.

\section{A uniform general signed rank test}\label{sec:intro_test}

We now define a general class of uniform signed rank tests that operate on a
family of related test statistics $\brace{T_n(x)}_{x \in (0,1)}$. Informally,
our test rejects when {any} test statistic in the family lies above a
corresponding modified critical value. These critical values are chosen to
correct for multiplicity by taking advantage of the structure of the family of
test statistics. The uniform nature of our test yields advantages in terms of
design sensitivity, which we describe in \cref{sec:uniform_design_sensitivity}.

For any $\varphi: (0,1) \to [0, \infty)$, define the family of test
statistics $\brace{T_n(x)}_{x \in (0,1)}$ by $T_n(x) = 0$ for
$x < 1/(n+1)$, and for $x \geq 1/(n+1)$,
\begin{align}
T_n(x)
  = \sum_{i = \ceil{(1-x)(n+1)}}^n
    \varphi\eparen{\frac{i}{n+1}} 1_{Y_{(i)} >0}
  = \sum_{i = \ceil{(1-x)(n+1)}}^n c_i 1_{Y_{(i)} >0},
  \label{eqn:ugsr_def}
\end{align}
where we have defined $c_i = \varphi\ebrace{{i}/({n+1})}$ for
convenience. For each $x$, $T_n(x)$ is a general signed rank statistic
using the truncated score function
$\varphi_x(q) = \varphi(q) \indicator{q \geq 1 - x}$, which, roughly speaking, gives weight zero to pair differences
whose absolute magnitudes lie below the $1-x$ quantile.  More precisely, it gives weight zero to pair differences below the $(1-x)(n+1)/(n-1) - 1/(n-1)$ quantile, hence the condition $T_n(x) = 0$ for $x < 1/(n+1)$.
There are $n$ distinct
nontrivial test statistics in this family, $T_n\brace{k/(n+1)}$ for
$k = 1, \dots, n$, corresponding to the partial sums
$\sum_{i=k}^n c_i\indicator{Y_{(i)} > 0}$ for $k = n, n -1, \dots, 1$. Thus the
family corresponds to a random walk with $n$ steps and step sizes determined by
the function $\varphi(\cdot)$.

Despite the generality of our construction in terms of the score
function $\varphi$, our family always consists of truncated versions of the full
test statistic. Such truncated statistics focus on subsets of the experimental
sample with large observed effects $\abs{Y_i}$. As such, our test will tend to
perform especially well against alternatives with large, rare effects.

Our uniform test will be characterized by a threshold function
$f_{\alpha,n}(x)$, the functional analogue of a critical value. Our test rejects
whenever $T_n(x) \geq f_{\alpha,n}(x)$ for any $x \in (0,1)$. As in the
fixed-sample case, there is a single worst-case distribution under $H_0(\Gamma)$
which maximizes the probability of rejection. We defer the proof of
\cref{th:worst_case} to \cref{sec:prove_lemmas},
along with most other proofs in this paper.
\begin{proposition}\label{th:worst_case}
  Fix any threshold function $f_{\alpha,n}: (0,1) \to \R_{> 0}$. Among all
  distributions in $H_0(\Gamma)$, the rejection probability
  $\text{pr}\condbrace{T_n(x) \geq f_{\alpha,n}(x) \text{ for some } x \in (0,1)}{\Fcal}$ is
  maximized when $\text{pr}\condparen{Y_i > 0}{\Fcal} = \Gamma / (1 + \Gamma)$ for all
  $i \in [n]$.
\end{proposition}
Under this worst-case distribution in $H_0(\Gamma)$, each step of the random
walk equals $c_i$ with probability $\rho_\Gamma = \Gamma / (1 + \Gamma)$
and zero otherwise; these steps are independent. The resulting mean and variance
of $T_n(x)$ are
\begin{align}
\mu_n(x) = E \ebrace{T_n(x)}
  &= \rho_\Gamma \sum_{i = \ceil{(1-x)(n+1)}}^n c_i  \nonumber\\
\sigma_n^2(x) = \text{var}\ebrace{T_n(x)}
  &= \rho_\Gamma (1 - \rho_\Gamma) \sum_{i = \ceil{(1-x)(n+1)}}^n c_i^2.
    \label{eq:sigmasq_defn}
\end{align}
Our threshold function requires a tuning parameter $x_0 > 0$ to be fixed in
advance, such that $\sigma_n^2(x_0) > 0$. If  $\sigma_n^2(x) = 0$ for all $x$,
then we cannot choose a valid $x_0$, but in this case, $T_n(x) = 0$
with probability one for all $x$ and the test statistic will be invariant to the data. We then
construct the following high-probability uniform upper boundary on the random
walk $T_n(x)$:
\begin{align}
f_{\alpha,n}(x) =
  \frac{1}{\lambda_n} \ebracket{
    \log\pfrac{1}{\alpha} +
    \sum_{i=\ceil{(1-x)(n+1)}}^n
      \log\ebrace{1 + \rho_\Gamma (e^{c_i \lambda_n} - 1)}
  },
\label{eq:uniform_bound}
\end{align}
where $\lambda_n = \left\{2 \log(\alpha^{-1}) / \sigma_n^2(x_0)\right\}^{1/2}$. For
notational simplicity, we omit the dependence of $f_{\alpha,n}$ on $x_0$.

\begin{theorem}\label{th:uniform_bound}
  Under $H_0(\Gamma)$, for any $x_0 > 0$ such that $\sigma_n^2(x_0) > 0$ and any
  $\alpha \in (0,1)$, 
  \begin{align*}
    \text{pr}\condbrace{T_n(x) \geq f_{\alpha,n}(x) \textrm{ for some } x \in (0,1)}{\Fcal}
    \leq \alpha.
\end{align*}
	\end{theorem}

\Cref{th:uniform_bound} justifies rejecting the sensitivity null $H_0(\Gamma)$
whenever $T_n(x) \geq f_{\alpha,n}(x)$ for some $x \in (0,1)$, allowing us to
adaptively choose a value of $x$ after seeing the data, while retaining Type I
error control at level $\alpha$. We call this test a {uniform general
  signed rank test}.

  The idea is illustrated well by considering the sign test score function $\varphi(q) = 1$ for which the resulting
truncated score functions $\varphi_x$ are exactly the score functions used in
Noether's test \citep{noether_simple_1973,rosenbaum_exact_2012}.
For this choice of $\varphi$, each $T_n(x)$ is the count of successes in a series of $\lfloor n(1-x) \rfloor$ Bernoulli trials or coin flips, each corresponding to a particular matched pair.  When $\Gamma >1$, the coins may be unfair; if we consider the most extreme distribution in $H_0(\Gamma)$ in the sense of \Cref{th:worst_case}, each coin has success probability $p_\Gamma$ and $T_n(x)$ is binomial.  Rather than waiting to observe all $n$ coin flips, or some fixed number of coin flips, before testing the fairness of the coins, our uniform test compares the number of successes to the testing threshold $f_{\alpha,n}$ after each successive coin flip and rejects if the threshold is ever crossed.  Because
the probability bound in \cref{th:uniform_bound} holds uniformly over all $x$,  the test controls Type I error properly. 
  More generally, we can think of the uniform rank test as simultaneously conducting
general signed rank tests with truncated score functions
$\varphi_x(q) = \varphi(q) \indicator{q \geq 1 - x}$ for all values
$x = 1/(n+1), \cdots, n/(n+1)$, but with modified critical values
$f_{\alpha,n}(x)$, and choosing the value of $x$ that yields the strongest inference.  The critical value $f_{\alpha,n}(x)$ is larger than the
fixed-sample exact critical value $c_{\alpha,n}(\Gamma)$ from
\cref{sec:general_sign_rank}, accounting for the uniformity of our test.

To give some intuition for the bound $f_{\alpha,n}$, we show that the following
function yields a good asymptotic approximation to $f_{\alpha,n}$ for large $n$:
\begin{align}
  g_{\alpha,n}(x) =
    \mu_n(x)
    + \ebrace{1 + \frac{\sigma_n^2(x)}{\sigma_n^2(x_0)}}
    \ebrace{\frac{\sigma_n^2(x_0) \log \alpha^{-1}}{2}}^{1/2}.  
\label{eq:g_fn}
\end{align}
Indeed, \cref{th:g_asymp} in \cref{sec:f_asymp_proof} shows that
$f_{\alpha,n}(x) = g_{\alpha,n}(x) + \Ocal(1)$ as $n \to \infty$. The leading
term in \eqref{eq:g_fn}, $\mu_n(x)$, is $\Ocal(n)$ and accounts for the drift of
the random walk. The next term is $\Ocal(\sqrt{n})$ and accounts for the
deviations of the random walk about its mean. As discussed in
\cref{sec:f_asymp_proof}, the parameter $x_0$ determines the value of $x$ for
which the boundary $g_{\alpha,n}(x)$ is optimized, and this motivates the choice
of $\lambda_n$ in the definition of $f_{\alpha,n}$. \Cref{th:uniform_bound}
would continue to hold with any choice $\lambda_n > 0$, but our choice enhances
interpretation by linking $\lambda_n$ to a point $x_0 \in (0,1)$ at which we
desire the bound to be especially small. \Cref{th:g_asymp} also shows that
$f_{\alpha,n}(x) \leq g_{\alpha,n}(x)$ for all $n$ and $x$, so that
$g_{\alpha,n}(x)$ yields a conservative threshold function with a simpler
analytical form, but the resulting test has slightly less power. We prove
\cref{th:uniform_bound} in \cref{sec:proofthm1} using a technique closely related to the
classical Cram\'er-Chernoff method (\citealp{cramer_sur_1938};
\citealp{chernoff_measure_1952}; \citealp{boucheron_concentration_2013}, \S 2.2;
\citealp{howard_exponential_2018}).

\section{Design sensitivity of the uniform test}
\label{sec:uniform_design_sensitivity}

We have shown that the uniform test may be thought of as simultaneously
conducting general signed rank tests at all values of $x$ with modified critical
values $f_{\alpha,n}(x)$. We might equivalently think of this as adjusting the
significance level $\alpha$ downwards, and to different values for different
$x$, in computing critical values for a sequence of general signed rank
tests in such a way that the familywise error rate is controlled over all 
the individual tests.  Compared to more general methods for adjusting 
multiple tests to control familywise error rates, such as the 
Bonferroni and Holm procedures, our test is expected to achieve 
tighter bounds because it explicitly accounts for the strong dependence among the test statistics.

Recalling that the design sensitivity of a general signed rank test
\eqref{eq:fixed_design_sensitivity} does not depend on $\alpha$, we may wonder
if the uniform test has design sensitivity equal to the supremum of the design
sensitivities of the test statistics $T_n(x)$, $x \in (0,1)$. This conclusion is
not quite trivial, since the adjusted significance levels in the uniform test
vary as $n$ grows. Nonetheless, it turns out to be true. We prove this for score
functions $\varphi: (0,1) \to [0,\infty)$ satisfying the following properties:

\begin{assumption} $\int_0^1 \varphi^2(x) \d x < \infty$.
\label{assmp:integrable}
\end{assumption}
\begin{assumption}
 $\varphi$ is discontinuous on a set of Lebesgue measure zero.
 \label{assmp:as_continuous}
\end{assumption}
\begin{assumption}
There exists a constant $a \in [0, 1 / 2)$ such that $\varphi$ is
  nonincreasing on $(0, a)$, nondecreasing on $(1 - a, 1)$, and bounded on
  $(a, 1 - a)$.
  \label{assmp:monotonish}
  \end{assumption}
\begin{assumption} $\int _{1-x}^1 \varphi(y) \d y > 0$ for all $x > 0$.
\label{assmp:nonzero}
\end{assumption}

\begin{theorem}\label{th:unif_ds}
  Suppose $\varphi$ satisfies Assumptions \ref{assmp:integrable}--\ref{assmp:nonzero} above, and $G$ is
  continuous. Then the design sensitivity of the corresponding uniform general
  signed rank test under $H_1(G)$ is
  \begin{multline*}
    \Gammatw_{\varphi,\unif} = \sup_{x \in (0,1)} \Gammatw(x)
      = \sup_{x \in (0,1)} \frac{\pi(x)}{1 - \pi(x)},
 \quad
    \pi(x) =
    \frac{
      \int_0^\infty \varphi\brace{G(y) - G(-y)} \indicator{G(y) - G(-y) \geq 1 - x}
      \d G(y)
    }{\int_{1-x}^1 \varphi(y) \d y}.
  \end{multline*}
\end{theorem}
The quantity $\pi(x)$ in \cref{th:unif_ds} is equivalent to the quantity $\pi$
in \eqref{eq:fixed_design_sensitivity} when the score function $\varphi$ is
replaced by the truncated score function
$\varphi_x(q) = \varphi(q) \indicator{q \geq 1 - x}$. Recall from
\eqref{eq:fixed_design_sensitivity} that $\pi$ is monotonically related to the
design sensitivity of the fixed-sample test by the relationship
$\Gamma = \pi / (1 - \pi)$ \citep{rosenbaum_design_2010-1}. Hence
$\Gammatw(x) = \pi(x) / \brace{1 - \pi(x)}$ gives the design sensitivity of a
fixed-sample test with a truncated score function $\varphi_x$, justifying the
interpretation of \cref{th:unif_ds} as showing that the design sensitivity of
the uniform test is equal to the supremum of the design sensitivities of the
family of fixed-sample tests with truncated score functions.

Most of the work in the proof of \cref{th:unif_ds} is captured by the following
pair of lemmas, both proved in \cref{sec:prove_lemmas}. The first characterizes the asymptotic behavior of the boundary
$f_{\alpha,n}(x)$ as $n \to \infty$, and the second generalizes a result of
\citet{sen_convergence_1970}.

\begin{lemma}\label{th:f_approximation}
  If $\varphi$ satisfies Assumptions \ref{assmp:integrable}--\ref{assmp:monotonish} above, then for any $x_0 > 0$ such
  that $\sigma_n^2(x_0) > 0$, any $\alpha \in (0,1)$, and any $x \in (0,1)$, we
  have $n^{-1} \mu_n(x) \to \rho_\Gamma \int_{1-x}^1 \varphi(y) \d y$ and
  $f_{\alpha,n}(x) = \mu_n(x) + \Ocal(\sqrt{n})$ as $n \to \infty$.
\end{lemma}


\begin{lemma}\label{th:T_convergence}
  If $\varphi$ satisfies Assumptions \ref{assmp:integrable}--\ref{assmp:nonzero}, and $Y_1,Y_2,\dots$ are drawn
independently \ from the same continuous distribution $G$, then with probability one,
  \begin{align*}
    \lim_{n \to \infty} \frac{1}{n}
      \sum_{i=1}^n \varphi\pfrac{i}{n+1} \indicator{Y_{(i)} > 0}
    = \int_0^\infty \varphi\ebrace{G(y) - G(-y)} \d G(y) 
  \end{align*}
\end{lemma}

\begin{proof}[of \cref{th:unif_ds}]
  Let $H(x) = G(x) - G(-x)$ denote the distribution of $\abs{Y}$. Fix
  any $x \in (0, 1)$. Applying \cref{th:T_convergence} to the truncated score
  function $\varphi_x(q) = \varphi(q) \indicator{q \geq 1 - x}$ yields
  \begin{align}
    \lim_{n \to \infty} \frac{T_n(x)}{n} =
    \int_{0}^\infty \varphi\ebrace{H(y)} \indicator{H(y) \geq 1 - x} \d G(y)
     \label{eq:T_limit}
  \end{align}
  with probability one.
  Meanwhile, \cref{th:f_approximation} implies that
  \begin{align}
    \lim_{n \to \infty} \frac{f_{\alpha,n}(x)}{n}
      = \rho_\Gamma \int_{1-x}^1 \varphi(y) \d y. \label{eq:f_limit}
  \end{align}
  Combining \eqref{eq:f_limit} with \eqref{eq:T_limit}, we conclude that
  \begin{multline*}
    \text{pr}\ebrace{T_n(x) \geq f_{\alpha,n}(x)}
      = \text{pr}\ebrace{n^{-1} T_n(x) \geq n^{-1} f_{\alpha,n}(x)}
      \to 1 \\
    \text{if }
    \int_{0}^\infty \varphi\ebrace{H(y)} \indicator{H(y) \geq 1 - x} \d G(y)
    > \rho_\Gamma \int_{1-x}^1 \varphi(y) \d y,
  \end{multline*}
  that is, if $\Gamma < \pi(x) / \brace{1 - \pi(x)}$. Since the uniform test
  rejects whenever $T_n(x) \geq f_{\alpha,n}(x)$ for some $x$, it will reject
  with probability approaching one whenever
  $\Gamma < \pi(x) / \brace{1 - \pi(x)}$ for some $x \in (0, 1)$. By a similar
  argument, $\text{pr}\brace{T_n(x) \geq f_{\alpha,n}(x)} \to 0$ if
  $\Gamma > \pi(x) / \brace{1 - \pi(x)}$, so the uniform test will reject with
  probability approaching zero if $\Gamma > \pi(x) / \brace{1 - \pi(x)}$ for all
  $x \in (0,1)$. The conclusion follows.
\end{proof}

Compare \cref{th:unif_ds} to Proposition 1 of
\citet{rosenbaum_exact_2012}. Rosenbaum constructs an adaptive test choosing
between two test statistics and achieving design sensitivity equal to the
maximum of the two component tests. \cref{th:unif_ds} shows that this principle
may be extended to an infinite family of tests, in this case because the family
possesses a dependence structure that allows us to construct an appropriate
uniform bound.

All of the score functions introduced in \cref{sec:general_sign_rank} satisfy
Assumptions \ref{assmp:integrable}--\ref{assmp:nonzero}. Most of these are
obvious; it suffices to show that the normal score function satisfies Assumption
\ref{assmp:integrable}.
\begin{proposition}\label{th:normal_scores_lp}
  For the normal score function, $\varphi(q) = \Phi^{-1}\brace{(1 + q) / 2}$, we
  have $\int_0^1 \varphi^p(x) \d x < \infty$ for all $p \geq 1$.
\end{proposition}

\begin{figure}
  \includegraphics{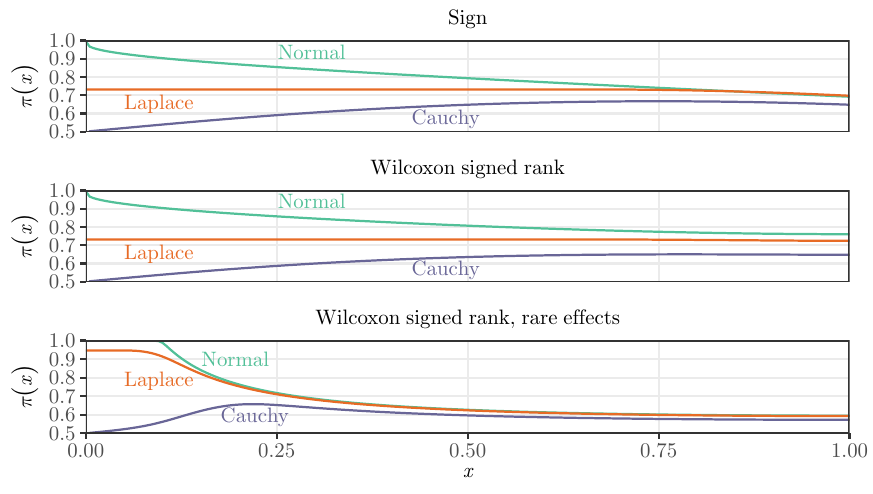}
  \caption{$\pi(x)$ from \cref{th:unif_ds} for sign and Wilcoxon signed rank score functions
    when $G$ is standard normal, Laplace (double exponential) or Cauchy. First
    two panels show alternative with $\tau = 1/2$. Bottom panel shows rare
    effects model: 90\% of pairs have no treatment effect, $\tau = 0$ while 10\%
    of pairs have a large treatment effect, $\tau = 5$. See
    \cref{fig:pi_plot_extra} in \cref{sec:addplots} for the corresponding plot with the normal score function, which has $\pi(x)$ qualitatively similar to
    that for the Wilcoxon signed rank score function. \label{fig:pi_plot}}
\end{figure}

\Cref{fig:pi_plot} plots $\pi(x)$ as defined in \cref{th:unif_ds}, showing how
the design sensitivity of a fixed-sample test varies with the truncation level
$x$. Recall that $\pi(x) = \Gammatw(x) / [1 + \Gammatw(x)]$, so $\pi(x)$ indicates
design sensitivity on a $[0,1]$ scale rather than a $[0,\infty)$ scale. Each
panel includes three alternative distributions $G$: normal with unit variance,
Laplace (double exponential) with unit scale, and Cauchy with unit scale. In the
first two panels, each distribution is centered at $\tau = 1/2$. The bottom
panel shows a rare effects model in which $G$ is a mixture of two of the given
base distributions, one centered at zero receiving 90\% of the total mass, and
the other centered at $\tau = 5$ receiving 10\% of the total mass. This
simulates a situation in which 90\% of pairs have no treatment effect, while the
remaining 10\% of pairs have a large constant treatment effect, so that the
average treatment effect remains equal to $1/2$.

The first two panels of \cref{fig:pi_plot} show $\pi(x)$ for the sign and Wilcoxon signed rank
score functions introduced in \cref{sec:general_sign_rank}; \cref{fig:pi_plot_extra} in \cref{sec:addplots} plots 
$\pi(x)$ for the normal score function, which is
qualitatively similar to $\pi(x)$ for the Wilcoxon signed rank score function. For the sign test, $\pi(x)$ is
maximized at some value $x < 1$ under all distributions, although the increase
is modest for the Laplace and Cauchy alternatives. This illustrates the benefits
of truncation with the sign test. With the Wilcoxon signed rank test, we still see dramatic gains
under a normal alternative, and indeed $\pi(x) \uparrow 1$ as $x \downarrow 0$
for all of our score functions under a normal alternative. This indicates we can
achieve infinite design sensitivity under normal tails, a fact that we prove in
\cref{th:normal_ds}. Under the Laplace or Cauchy alternatives, however, we do
not see substantial gains in $\pi(x)$ as $x$ decreases from one for the Wilcoxon signed rank test;
the same holds true for the normal score function. Under the heavier-tailed Laplace and Cauchy alternatives, it seems,
score functions that place more weight on larger outcomes do not benefit from
narrowing attention to a subset of pairs with the largest absolute
differences. Informally speaking, the higher likelihood of large outliers means
less information is present in the tails.

The $\pi(x)$ functions in the bottom panel, computed under a rare effects model,
tells a different story. Here, a uniform Wilcoxon signed rank test benefits from narrowing attention
to a subset of pairs with large absolute differences regardless of the
alternative distribution, although gains are still more modest for the Cauchy
alternative than for the others. This confirms that, when effects are large and rare, a test which restricts attention accordingly will retain high power under more adversarial conditions and achieve a larger design sensitivity than a less-focused test.

\Cref{fig:pi_plot} makes it clear that the best choice of $x$ depends on the
alternative distribution $G$ and the score function in a complicated manner. The
advantage of our uniform test is that it can adapt to the alternative at hand
without prior knowledge, achieving performance equivalent to the oracle choice
of $x$ in terms of design sensitivity. It it also notable that all four score
functions exhibit identical behavior near $x = 0$. The following result makes
this observation precise whenever $G$ is continuous with infinite support. We
show that the limiting behavior of $\pi(x)$ as $x \downarrow 0$ is often
determined by the tails of $G$ alone, not by the score function $\varphi$, and
this may be used to bound the design sensitivity from below over a broad class of
score functions.

\begin{theorem}\label{th:limiting_ds}
  Suppose $\varphi$ satisfies Assumptions \ref{assmp:integrable}--\ref{assmp:nonzero} above, and suppose $G$ has
  positive density $g(x)$ with respect to Lebesgue measure for all $x \in \mathbb{R}$. Then 

  \begin{align}
    \Gammatw_{\varphi,\unif} \geq
      \liminf_{q \uparrow \infty} \frac{g(q)}{g(-q)}. \label{eq:ds_limit}
  \end{align}
\end{theorem}

Plugging the normal density into \cref{th:limiting_ds} for $g(x)$ confirms the
fact suggested by \cref{fig:pi_plot}:

\begin{corollary}\label{th:normal_ds}
  If $G = \Normal(\tau, \sigma^2)$, then $\Gammatw_{\varphi,\unif} =
  \infty$. That is, no matter what value of $\Gamma$ is used in a sensitivity
  analysis with a uniform general signed rank test, the power under $H_1(G)$
  tends to one as $n \to \infty$.
\end{corollary}

As also suggested by \cref{fig:pi_plot}, the uniform test does not achieve
infinite design sensitivity for other distributions. The limit
\eqref{eq:ds_limit} is equal to $\exp(2\mu/s)$ for the Laplace distribution with
mean $\mu$ and scale $s$, while it is equal to one for the Cauchy distribution,
hence \cref{th:limiting_ds} offers no meaningful guarantee for the Cauchy
distribution.

\section{Simulations}\label{sec:simulations}

\Cref{fig:fixed_vs_uniform} illustrates
\cref{th:unif_ds} with simulations under standard normal, Laplace and Cauchy
alternatives; in each case $\tau = 1/2$, except for the rare effects panels
in \cref{fig:fixed_vs_uniform}, which uses the rare effects model described in
\cref{sec:uniform_design_sensitivity}. We simulate both fixed-sample
tests and uniform tests based on \cref{th:uniform_bound} with the three score
functions introduced in \cref{sec:general_sign_rank}. All tests are run with
level $\alpha = 0.05$ and plots are based on 10,000 replications.  
\begin{figure}
\includegraphics{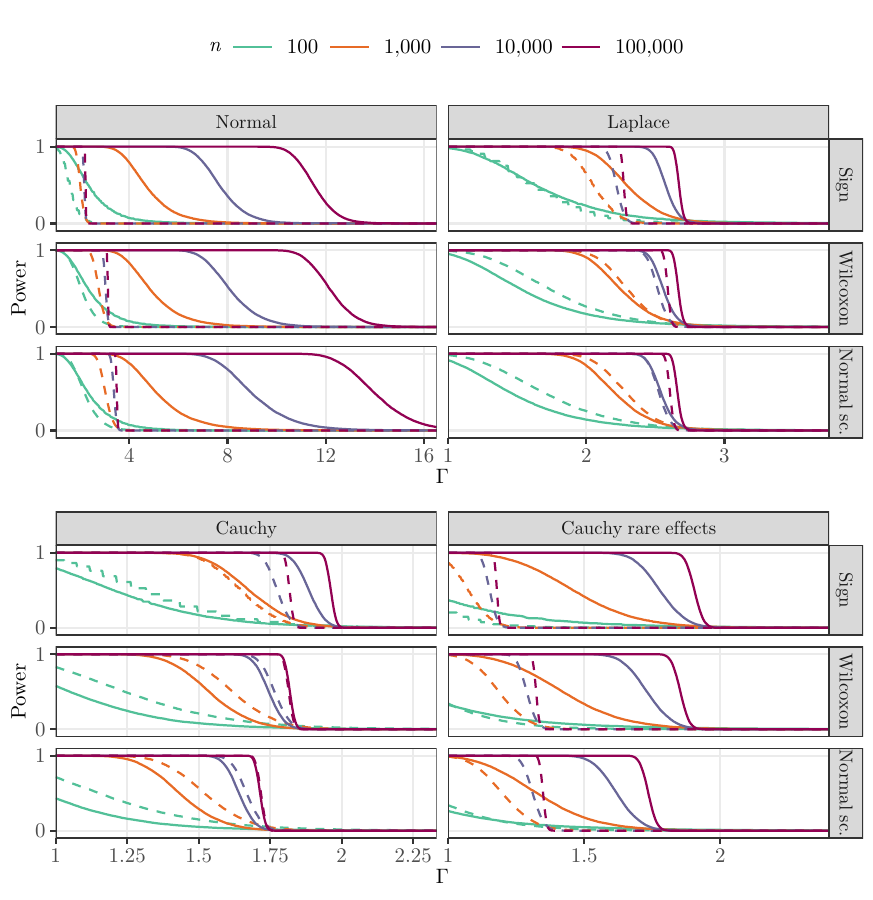}
\caption{Comparison of simulated power for fixed-sample tests (dashed lines) vs.
  uniform tests (solid lines), based on 10,000 replications. Cauchy rare
  effects panels show the rare effects alternative model based on Cauchy
  distribution, as described in \cref{sec:uniform_design_sensitivity}. Other
  panels show alternative model $H_1(G)$ with distribution $G$ as indicated,
  having center $1/2$ and unit scale. All tests use $\alpha =
  0.05$. \label{fig:fixed_vs_uniform}}
\end{figure}

Our results show that the improvements in performance for uniform tests expected from \cref{fig:pi_plot} are in many cases realized in moderately-sized finite samples.
\Cref{fig:fixed_vs_uniform}
compares power for each uniform test to the the corresponding fixed-sample test
based on the same score function, for a variety of different sample sizes. In the normal case, the uniform test does not
indicate finite design sensitivity, as we expect from \cref{th:normal_ds}, and
all uniform tests show substantial gains over their fixed-sample counterparts
for $n \geq$ 1,000. In the Laplace and Cauchy cases, the uniform sign test still
shows gains, but uniform tests based on other score functions often fail to
outperform their fixed-sample counterparts, as we expect from
\cref{fig:pi_plot}; only at the very largest sample sizes do the uniform tests 
remain competitive in nearly all cases. However, in the rare effects case 
each uniform test improves substantially on its fixed-sample counterpart for
$n \geq 1,000$, even with Cauchy noise. Though not shown here, the gains for
normal and Laplace noise under the rare effects model are even more
dramatic. \Cref{fig:compare_scores} in \cref{sec:addplots} gives
additional power comparisons, while \cref{sec:typeI} contains simulation results
under the null, illustrating the true Type I error rates of our uniform test in
some representative settings.

Overall, the uniform sign test shows considerable promise for use
in practice. It is competitive in all cases and is the strongest performer of
the four tests considered here in a number of cases. This is particularly
interesting since the fixed-sample sign test is arguably the least attractive
among the fixed-sample tests we have considered. It seems the landscape of
uniform general signed rank tests is qualitatively different from that of their
fixed-sample counterparts.

\section{Application: impact of fish consumption on mercury concentration}
\label{sec:data}

Mercury can be harmful to human health when concentrated too heavily in the
bloodstream, and evidence suggests that consuming large amounts of fish can lead
to elevated levels of mercury in the blood \citep{mahaffey2004blood}.  To study
the impact of a high-fish diet on mercury concentration we use data from the
National Health and Nutrition Examination Survey \citep{nchs2017nhanes}, which
records information about respondents' diets as well as analysis of blood
samples, including a measure of total mercury concentration.  Each of the 1,672
respondents from 2007 to 2016 who consumed an average of 15 or more servings of
fish monthly was matched to a similar respondent consuming two or fewer servings
per month.  Respondents were matched within the same two-year period, and pairs
were chosen by optimal matching on a robust Mahalanobis distance \citep[\S
8.3]{rosenbaum_design_2010} computed from respondent age, household income,
gender, ethnicity, cigarettes smoked per day, and indicators for high school
graduation, missing high school graduation status, and smoking more than 7
cigarettes per
day.  
A propensity score was fit to these same variables, and matches were required to
obey a propensity score caliper of 0.2 standard deviations
\citep{rosenbaum1985constructing}.  The final matched sample of 1,672 pairs
achieved close balance on covariates; see \cref{tab:balance} in \cref{sec:addplots}.  For matching, we used R packages \texttt{rcbalance} and
\texttt{optmatch} with package \texttt{cobalt} used for balance checking
\citep{pimentel2016large, hansen2006optimal, greifer2018cobalt}.  For more
discussion on the optimal construction of matched samples see
\citet{rosenbaum1989optimal}, \citet{hansen2004full},
\citet{zubizarreta2014matching}, and \citet{pimentel2015large}.

Although the balance on observed variables in Table \cref{tab:balance} is
very close, individuals with high-fish diets may differ from individuals with
low-fish diets on many unobserved attributes correlated with mercury levels.
Accordingly, we are interested not only in whether a test assuming an absence of
unobserved confounders rejects the null hypothesis, but in how sensitive such a
result is to potential bias from unobserved confounders.

In each of the 1,672 pairs formed, we computed the difference in total mercury
concentration, in micrograms per mole, between the respondent with the high-fish
diet and the respondent with the low-fish diet.  The average concentrations for
matched individuals with high-fish diets and low-fish diets were 3.76 and 1.02
respectively, yielding an average pair difference of 2.73 micrograms per
mole. We next tested the sharp null of no effect of treatment in any pair.
Mercury measurements were rounded to two decimal places leading to some ties, so
for each test we used a minor modification of our procedure described in
\cref{sec:ties} to account for ties. We set
$x_0 = 1/3$ for the uniform bound. Setting $x_0$ too large leads to slight
conservatism, as shown in \cref{fig:typeI_x0} in \cref{sec:typeI},
while setting $x_0$ too small optimizes for a test using only a small fraction
of the data. Our choice simply balances between these two extremes. In other
settings, prior information about the distribution of outcomes and treatment
effects may inform the choice of $x_0$ (see, e.g., \cref{fig:pi_plot}). The
methods and data used for this analysis will be made available to readers via
the R package \texttt{uniformrank}.

The first three columns of \Cref{tab:sens_all} show the results of sensitivity
analysis in the matched data for the three general signed rank tests considered
in this paper.  For each of these test statistics, the na\"{i}ve test with
$\Gamma = 1$ produces results highly significant at the 0.05 level.  The
numbers in the table describe the smallest amount of unmeasured bias necessary
to explain the observed effects assuming there is no true effect of
treatment, that is, the minimum value of $\Gamma$ at which we fail to reject
the sensitivity analysis null.  For example, the fixed-sample sign test ceases
to reject the null when we allow for an unobserved confounder that increases
the odds of a high-fish diet by a factor of $\Gamma = 4.82$; in contrast, the
uniform sign test requires an unobserved confounder that increases the odds of
a high fish diet by $\Gamma = 10.51$ before it ceases to reject.

\begin{table}[ht]
\centering
\begin{tabular}{l|rr|rr}
\hline
&\multicolumn{2}{c|}{1,672 pairs} & \multicolumn{2}{c}{190 pairs} \\
Score function & Fixed-sample & Uniform & Fixed-sample & Uniform \\
\hline
Sign & 4.82 & 10.51 & 3.72 & 8.29 \\ 
  Wilcoxon Signed Rank & 8.06 & 10.47 & 6.04 & 8.09 \\ 
  Normal Scores & 8.55 & 10.36 & 6.52 & 7.95\\ 
\hline
\end{tabular}

\caption{Sensitivity analysis for matched data.  Each cell of the table
  represents a different test statistic for testing the null of no effect of a
  fish diet on mercury concentration; the first two columns give results for the
  full matched sample of 1,672 pairs, while the third and fourth columns give
  results for the smaller sample from 2015-2016 alone.  The number in each cell
  is the smallest degree of unmeasured confounding $\Gamma$ necessary in the
  sensitivity analysis model before the test no longer rejects at the
  $\alpha = 0.05$ level. \label{tab:sens_all}}
\end{table}

Repeating the same test many times with different test statistics, as
in \cref{tab:sens_all}, may lead to problems with Type I error control.  In practice one should select a single test
statistic in advance, possibly based on a pilot sample
\citep{heller_split_2009}.  We show the results of several tests here to
illustrate the impact of the choice of test statistic and complement the
discussion in \cref{sec:simulations}.

Several interesting patterns are clear in the full-sample results of
\cref{tab:sens_all}.  First, regardless of the score function used, the uniform
version of the test is less sensitive to bias from unmeasured confounding than the fixed-sample
version.  This pattern is consistent with \cref{th:unif_ds}, which tells us that
in large samples the uniform test should perform at least as well as any
fixed-sample test it incorporates.  Second, the performance of the uniform test
across score functions varies much less than the performance of the fixed-sample
version across score functions.  In particular, the sign test performs
substantially worse than any other test examined in the fixed-sample case, but
is comparable to the other score functions in
the uniform setting, corroborating the evidence from simulations in
\cref{sec:simulations}. In this dataset, as in the simulations, adapting over
many different truncated statistics appears to compensate for the deficiencies
of the fixed-sample sign test.

Finally, we consider the role of sample size by analyzing the
subset of the matched dataset consisting of respondents from the
final two-year period (2015-2016), a total of 190 pairs.  The final two columns
of \cref{tab:sens_all} repeat the analysis for this smaller dataset.  The same
pattern of results is observed, with the uniform test outperforming the
fixed-sample test for each score function, and the sign test performing best
among uniform tests.  Although the benefits of uniform testing articulated in
\cref{th:unif_ds} relate to asymptotic performance in large samples, uniform
tests may also offer substantial improvement in datasets of only moderate size.

\section{Conclusion and future work}\label{sec:conclusion}


We have focused on continuous outcomes, although minor modifications, described in \cref{sec:ties}, can handle cases when ties  are present but rare as in \cref{sec:data}.  However, outcomes with relatively few unique
values require alternative methodology. In such cases, the random walk
$\brace{T_n^\star(x)}_{x \in (0,1)}$ for tied data defined in \cref{sec:ties} will have few steps, at most the
number of unique values of the outcome, with each step comprised of many
individual observations, namely all those pairs with absolute outcome equal to a
given value. In the sequential analysis literature, such random walks are
handled well by group sequential designs \citep{pocock_group_1977,
  obrien_multiple_1979, lan_discrete_1983, jennison_group_2000}. An application
to uniform general signed rank tests may yield promising future results.

Another interesting avenue is theoretical properties
besides design sensitivity for uniform general signed rank tests. 
\citet{rosenbaum2015bahadur} characterizes the Bahadur efficiency of 
sensitivity analyses based on signed rank tests, allowing comparison of 
asymptotic performance for different test statistics at values of $\Gamma$ 
strictly smaller than the design sensitivity.  Deriving Bahadur slopes for
uniform signed rank tests would shed light on their relative
performance at small values of $\Gamma$.
Additionally,
\citet[\S 6]{lehmann_testing_2005} discuss the locally most powerful
property of general signed rank tests against particular families of
alternatives determined by the function $\varphi$. The uniform test chooses
adaptively from among a family of related $\varphi$ functions, and the implications for local
optimality in the sense discussed by Lehmann and Romano deserve attention.

\section{Acknowledgments}

We thank Eli Ben-Michael for the conversation that sparked this project. Howard
thanks Office Of Naval Research (ONR) Grant N00014-15-1-2367.

This is a pre-copyedited, author-produced version of an article accepted for publication in Biometrika following peer review. The version of record:

Howard, S.R., and Pimentel, S.D. (2021). The uniform general signed rank test and its design sensitivity. Biometrika 108, 381-396.

is available online at: \url{https://academic.oup.com/biomet/article-abstract/108/2/381/5911093}.

\setlength{\bibsep}{0pt plus 0.3ex}

\bibliographystyle{agsm}
\bibliography{rank_uniform}

\appendix

\section{Proof of Theorem 1}\label{sec:proofthm1}

  Throughout the proof we condition on $\Fcal$ implicitly in probability
  statements. Let $S_i = \indicator{Y_{(i)} > 0}$ for $i \in [n]$, so that
  $T_n\brace{k/(n+1)} = \sum_{i=n+1-k}^n c_i S_i$ for each $k \in [n]$. By
  Proposition 1, under the worst-case distribution in $H_0(\Gamma)$,
  $(S_i)_{i \in [n]}$ are distributed as $n$ i.i.d.\ Bernoulli($\rho_\Gamma$)
  random variables. The moment-generating function of 
  $c_i S_i$ is
  \begin{align}
    E\left( e^{\lambda c_i S_i}\right) = 1 + \rho_\Gamma (e^{c_i \lambda} - 1)
    \label{eq:bernoulli_mgf}
  \end{align}
for all $ \lambda \in \R$.  Now define $(L_k)_{k=0}^n$ by $L_0 = 1$ and, for $k \in [n]$,
  \begin{align}
    L_k &= \exp\ebracket{
          \lambda_n T_n\pfrac{k}{n+1}
          - \hspace{-0.6em}\sum_{i=n+1-k}^n
            \log\ebrace{1 + \rho_\Gamma (e^{c_i \lambda_n} - 1)}}
      = \hspace{-0.6em} \prod_{i=n+1-k}^n
          \frac{e^{\lambda_n c_i S_i}}{1 + \rho_\Gamma (e^{c_i \lambda_n} - 1)}.
          \label{eq:Lk_defn}
  \end{align}
From \eqref{eq:bernoulli_mgf} and \eqref{eq:Lk_defn} we see that
  $E\condparen{L_k}{S_n, S_{n-1}, \dots, S_{n+2-k}} = L_{k-1}$, so $L_k$
  is a nonnegative martingale with respect to the natural filtration defined by
 $S_n, S_{n-1}, \dots, S_1$. By Ville's maximal inequality for
  nonnegative supermartingales (\citealp{ville_etude_1939};
  \citealp{durrett_probability:_2013}, Exercise 5.7.1),
  \begin{align*}
    \alpha &\geq \text{pr}\eparen{ L_k \geq \alpha^{-1}\text{ for some }k \in [n] } = \text{pr}\ebrace{
            T_n\pfrac{k}{n+1} \geq f_{\alpha,n}\pfrac{k}{n+1}\text{ for some }k \in [n] 
         } \\
      &= \text{pr}\ebrace{
           T_n(x) \geq f_{\alpha,n}(x)        \text{ for some } x
             \in \ebrace{\frac{1}{n+1},  \frac{2}{n+1}, \ldots,  \frac{n}{n+1}}} \\
      &= \text{pr}\ebrace{ T_n(x) \geq f_{\alpha,n}(x)   \text{ for some } x
             \in (0,1)}.
\end{align*}
The final equality follows since values
$x = 1/(n+1), 2/(n+1), \dots, n/(n+1)$ capture all distinct values of
both $T_n(x)$ and $f_{\alpha,n}(x)$ for $x \geq 1/(n+1)$; adding the region
$0 < x < 1/(n+1)$ does not change the overall probability since $T_n(x) = 0$
here and $f_{\alpha,n}(x)$ is strictly positive.

\section{Handling ties}\label{sec:ties}

Under the assumption that outcomes are drawn from a continuous distribution,
ties among outcome observations occur with probability zero.  In practice
however, tied outcome data may arise in a variety of settings.  In this section
we discuss how to adapt the results of the paper to the setting of ties.

Let $Y_{(1)}, \ldots, Y_{(n)}$ be the outcome data ordered in any way so that
$|Y_{(1)}| \leq |Y_{(2)}| \leq \ldots |Y_{(n)}|$.  This ordering is
not unique when ties are present; in such cases, choose one such ordering
arbitrarily. We may still apply the methods described in the paper directly to
conduct a test.  The test statistic and the uniform bound are clearly defined
given our chosen ordering of outcomes, and Theorem 1 holds since
no aspect of its proof depends on the absence of ties. We remark that it is
reasonable to expect $\text{pr}(Y_i = 0) > 0$ in the presence of ties; however, this
only reduces $\text{pr}\condparen{Y_i > 0}{\Fcal}$, so Proposition 1 and
Theorem 1 continue to hold.

However, the version of our uniform test in Theorem 1 depends on
the ordering we choose, perhaps arbitrarily, for $(Y_{(i)})$.  To remove this
undesirable feature of the procedure, we may instead use a generalization of
$T_n(x)$ that is invariant to the specific choice of ordering in the tied
setting. We write $T^\star_n(x)$ for this new test statistic.  The
intuition for $T^\star_n$ comes from recognizing that when ties are
present, one or more test statistics in the family
$\brace{T_n(x)}_{x \in (0,1)}$ are partial sums that include some terms with a
particular absolute value but exclude others with the same absolute value, and
that the scores associated with these terms may be different from one another.
We obtain the family $\brace{T^\star_n(x)}_{x \in (0,1)}$ by replacing the score
for each tied value by the average of scores for all indices involved in the
tie, and by adding all these terms together to the partial sum rather than
allowing partial sums that contain some terms but not all.

Formally, define
$\mathcal{J}_i = \ebrace{j \in [n]: \ \eabs{Y_{(j)}} = \eabs{Y_{(i)}}}$, the set
of ranks with equal absolute pair differences to the $i$th ranked pair. Let
$m(i) = \min \mathcal{J}_i$, the lowest rank within the tied group containing
the $i$th ranked pair. Now define the test statistic
\begin{align*}
T^\star_{\varphi}(x) =
  \sum_{\ebrace{i: m(i) \geq (1-x)(n+1)}}c_i^\star1_{Y_{(i)} > 0},
\qquad
c_i^\star = \eabs{\mathcal{J}_i}^{-1}
  \sum_{j \in \mathcal{J}_i}\varphi\eparen{\frac{j}{n+1}}.
\end{align*}

When a group of pairs share the same absolute outcome value, this test statistic
treats all these pairs as a single unit, including either all or none of them in
the partial sum, and assigning each a score equal to the average score across
all members in the tied set.  If there are only $k < n$ distinct
absolute outcome values, there are only $k$ distinct nontrivial values for
$T^\star_n(x)$; however, if no ties are present, $T^\star_n$ is
identical to $T_n$.

We obtain a uniform boundary for $T^\star_n(x)$ by substituting $c^\star_i$ for
$c_i$ in (3) and (4), yielding new
quantities ${\sigma^\star_n}^2(x)$ and $f^\star_{\alpha,n}(x)$. In the absence
of ties, the quantities ${\sigma^\star_n}^2$, and $f^\star_{\alpha,n}$ coincide
with the original quantities $\sigma_n^2$, and $f_{\alpha,n}$. However, the
quantities $(c_i^\star)_{i=1}^n$ are random, unlike $(c_i)$, hence
$\sigma_n^\star$, and $f^\star_{\alpha,n}$ are random as well. This requires no
real change to the analysis, since these quantities are $\Fcal$-measurable and
we condition on $\Fcal$ throughout the proof of Theorem 1. As the
reader may expect, the new boundary $f^\star_{\alpha,n}$ yields a valid uniform
test of the sensitivity null $H_0(\Gamma)$ using the order-invariant test
statistic $T_n^\star$.

\begin{theorem}\label{th:tie_bound}
  Under $H_0(\Gamma)$, for any $\Fcal$-measurable $x_0 > 0$ such that
  ${\sigma^\star_n}^2(x_0) > 0$ a.s., and any $\alpha \in (0,1)$, we have
  $\text{pr}\condbrace{T^\star_n(x) \geq
    f^\star_{\alpha,n}(x) \text{ for some }x \in (0, 1) }{\Fcal} \leq \alpha$.
\end{theorem}
\begin{proof}
  Write
  $\widetilde{T}_n(x) = \sum_{i=\ceil{(1-x)(n+1)}}^n c^\star_i
  \indicator{Y_{(i)} > 0}$; this is the same as $T_n(x)$ with $c^\star_i$
  substituted for $c_i$. Repeating the proof of Theorem 1 with
  $\sigma^\star_n$ and $f^\star_{\alpha,n}$ in place of their unstarred
  counterparts, we obtain
  \begin{align}
    \text{pr}\condbrace{
      \widetilde{T}_n(x) \geq f^\star_{\alpha,n}(x)\text{ for some }x \in (0, 1) 
    }{\Fcal} \leq \alpha. \label{eq:Ttilde_bound}
  \end{align}
  Since $m(i) \leq i$ and $c_i \geq 0$ for all $i$, we have
  $T^\star_n(x) \leq \widetilde{T}_n(x)$ for all $x$, which implies
  the result together with \eqref{eq:Ttilde_bound}.
\end{proof}

\section{The redescending score function}\label{sec:redescending}

In this section we discuss an additional, {redescending}
score function, defined as follows in terms of three positive integer parameters
$\underline{m} \leq \overline{m} \leq m$:
\begin{align}
\label{eqn:redescending}
\varphi(q) = \sum_{l=\underline{m}}^{\overline{m}} \frac{l}{m} \binom{m}{l}
  q^{l-1} (1-q)^{m-l}.
\end{align}
The function's name reflects its shape; it rises as $q$ increases from zero,
as the Wilcoxon signed rank score and normal score functions do, but falls
back to zero as $q$ approaches one, unlike the score functions in the main manuscript. The
resulting statistic puts more weight on pairs with larger absolute differences,
but excludes the most extreme outlying observations. In some cases it may be
preferable to remove potential outliers from the data in advance rather than to
use a redescending score function; for more discussion, see
\citet{huber_robust_2009}.

The score function for general $(m,\underline{m},\overline{m})$ is due
originally to \citet{rosenbaum_new_2011}, who introduced a family of test
statistics defined by three positive integers
$\underline{m} \leq \overline{m} \leq m$ where $m$ is no greater than the
overall sample size $n$.  Each statistic takes a sum over all size-$m$ subsets
of observations; within each subset, observations are ranked by absolute
magnitude and the number of positive differences among observations with ranks
$\underline{m}$ through $\overline{m}$ inclusive is added to the sum.  This
family includes the sign statistic ($m = \underline{m} = \overline{m} = 1$) and
a close approximation to the Wilcoxon signed rank statistic
($m = \underline{m} = \overline{m} = 2$) as special cases. These statistics are
general signed rank statistics, and although their exact score functions may
be computationally intensive and depend on sample size in general, Rosenbaum
gives \eqref{eqn:redescending} as a convenient asymptotic approximation to the
scores.  The redescending behavior described above occurs whenever
$\overline{m} < m$, and \citet{rosenbaum_weighted_2014} studied the specific
setting $(m,\underline{m},\overline{m}) = (20, 12, 19)$.

Using the redescending score function within the uniform general signed rank
test as we have defined it in \eqref{eqn:ugsr_def} is not problematic from a
technical perspective, but the resulting test runs contrary to the goal of the
redescending function.  In particular, the truncated score functions
$\phi(q)1_{q \geq 1-x}$ will eventually restrict attention to the most extreme
pair differences, which the redescending score function is designed to emphasize
less relative to other observations; this is in contrast to the other score
functions discussed in the manuscript, which are all nondecreasing.  A more
natural generalization of the truncation idea in the redescending case is to
truncate regions of the domain associated with lower score $\varphi(x)$ first.
We now show how to generalize the test statistic to truncate in this manner.

Let $D$ be a mapping between real numbers in the closed unit interval and
subsets of the unit interval such that $D(0) = \{\}$, $D(1) = [0,1]$, and
$D(x_1) \subseteq D(x_2)$ whenever $x_1 < x_2$.  We now reformulate the test
statistic and the uniform bound as functions of both $x$ and $D$:
\begin{align}
\label{eqn:ugsrDx_def}
T_n(x; D)
&= \sum_{\left\{i: \,\, {i}/{(n+1)} \in D(x)\right\}}
    \varphi\eparen{\frac{i}{n+1}} 1_{Y_{(i)} >0},\\
f_{\alpha,n}(x; D) &=
  \frac{1}{\lambda_n} \ebracket{
    \log\pfrac{1}{\alpha} +
    \sum_{\left\{i:\,\, {i}/{(n+1)} \in D(x)\right\}}
      \log\ebrace{1 + \rho_\Gamma (e^{c_i \lambda_n} - 1)}
  },
\label{eqn:uniformDx_bound}
\end{align}
with a similar modification to the sum in $\sigma^2_n(x_0)$ which appears in the
definition of $\lambda_n$. Intuitively, $D(x)$ defines an order of truncation
for the $n$ observed ranks.  Definitions \eqref{eqn:ugsr_def} and
\eqref{eq:uniform_bound} are special cases obtained by choosing
$D(x) = [1 - x, 1]$, which corresponds to truncating the observations with
smallest absolute ranks first.  For bounded redescending score functions such as
(\ref{eqn:redescending}), we suggest the following alternative choice of $D$:
\[
D(x) = \left\{q: \varphi(q) > (1-x)\max_{0 < q < 1}\varphi(q) \right\}.
\]
This $D$-function specifies that observations with the lowest scores will be trimmed first from the test.  For redescending score functions, this means specifically that the largest observations will be trimmed before certain intermediate observations, a desirable property in cases where the extreme observations may be unreliable.


We now discuss how to generalize the theory in the main manuscript to give
results for test statistic (\ref{eqn:ugsrDx_def}) and boundary
(\ref{eqn:uniformDx_bound}).  For any given sample, function $D(x)$ defines a
specific order of truncation among the $n$ observations.  Taking any
$i,j \in [n]$, $i \neq j$, we say that $i <_D j$, or $i$ is truncated later than
$j$, if there exists some $x \in (0,1)$ such that $i \in D(x)$ and
$j \notin D(x)$.  Note that $i <_D j$ and $j <_D i$ cannot both hold since
$D(x_1) \subseteq D(x_2)$ for any $x_1 < x_2$.  If neither $i <_D j$ nor
$j <_D i$ then $i \sim_D j$, meaning that $i$ and $j$ lie in the same
equivalence class with respect to truncation.

Generalizing Theorem 1 to hold for $T_n(x; D)$ and $f_{\alpha,n}(x;D)$ in place of $T_n(x)$ and $f_{\alpha,n}(x)$ is straightforward.  In fact, if no pair of observations $i,j \in [n]$ lies in the same equivalence class with respect to truncation, so that either $i <_D j$ or $j <_D i$ in all cases, then the arguments require no changes except for substituting the more general definitions, since $T_n(x;D)$, like $T_n(x)$, is a random walk with $n$ distinct partial sums.

If there exist one or more pairs of observations $i,j \in [n]$ such that
$i \sim_D j$, then $T_n(x;D)$ does not have $n$ distinct partial sums; since
some observations are always truncated together, there will only be $k$ distinct
partial sums for some $k < n$.  Since the proof of Theorem 1 relies on a random
walk with $n$ distinct partial sums, we construct a new random walk with $n$
distinct partial sums, $k$ of which are the partial sums of $T_n(x;D)$.  Let
$\pi:[n] \to [n]$ be a permutation of the indices such that $\pi(i) < \pi(j)$
whenever $i <_D j$.  Define
\begin{align*}
T_n(j;D) &= \sum_{\ebrace{i: \pi(i) \leq \pi(j)}}
  \varphi\left(\frac{i}{n+1}\right)1_{Y_{(i)} > 0},\\
f_{\alpha,n}(j;D) &= \frac{1}{\lambda_n} \ebracket{
    \log\pfrac{1}{\alpha} +
    \sum_{
    \ebrace{i: \pi(i) \leq \pi(j)}}
      \log\ebrace{1 + \rho_\Gamma (e^{c_i \lambda_n} - 1)}
  }.
\end{align*}
The dependence on $D$ here is through the restriction that the ordering $\pi$
respect the ordering given by $<_D$. Consider the set
$\mathcal{I} = \ebrace{i \in [n]: \pi(i) < \pi(j) \text{ or } j <_D i \text{ for
    all } j \in [n] \setminus i}$.  Set $\mathcal{I}$ contains exactly one
observation from each of the $k$ equivalence classes, choosing the first element
by ordering $\pi$ in equivalence classes with more than one element, and the
values $\lbrace T_n(i;D) : i \in \mathcal{I} \rbrace$ are exactly the $k$
distinct partial sums of $T_n(x;D)$.  Thus the partial sums of $T_n(x;D)$ form a
subset of the partial sums $\lbrace T_n(j;D) : j \in [n] \rbrace$.  In addition,
the values $\lbrace f_{n,\alpha}(i;D) : i \in \mathcal{I} \rbrace$ capture all
distinct values of $f_{\alpha,n}(x;D)$.

To finish the proof, it suffices to repeat our previous arguments for Theorem 1 replacing $T_n(x)$ and $f_{\alpha,n}(x)$ with $T_n(j;D)$ and $f_{n,\alpha}(j; D)$. Because the individual tests performed by $T_n(j;D)$ at different values $j$ are a superset of the individual tests performed by $T_n(x;D)$ at different values $x$, the rejection event for the latter uniform test is a subset of the rejection event for the former, and control of Type I error for the latter guarantees control of Type I error for the former.

Finally we comment on adaptations needed to generalize Theorem 2 for statistic (\ref{eqn:ugsrDx_def}) and (\ref{eqn:uniformDx_bound}).  We modify Assumptions 2--4 to impose conditions on $D$ as follows:

\stepcounter{assumption}

\begin{assumption} For all $x \in (0,1]$, $q \mapsto \varphi(q)1_{q \in D(x)}$ is discontinuous on a set of Lebesgue measure zero.
\label{assmp:as_cont_allX}
\end{assumption}
\begin{assumption}  For all $x \in (0,1)$, there exists a constant $a_x \in [0, 1 / 2)$ such that $q \mapsto \varphi(q)1_{q \in D(x)}$ is
  nonincreasing on $(0, a_x)$, nondecreasing on $(1 - a_x, 1)$, and bounded on
  $(a_x, 1 - a_x)$.
  \label{assmp:monotonish_allX}
\end{assumption}
\begin{assumption}
 $\int _{D(x)} \varphi(y) \d y > 0$ for all $x > 0$.
 \label{assmp:nonzero_allX}
\end{assumption}
Under these conditions, the design sensitivity of the uniform test based on (\ref{eqn:ugsrDx_def}) and (\ref{eqn:uniformDx_bound}) is now the supremum of $\pi(x;D)/[1 - \pi(x;D)]$ over $x \in (0,1)$, where
\[
\pi(x;D) =   
    \frac{
      \int_0^\infty \varphi\brace{G(y) - G(-y)} \indicator{G(y) - G(-y) \in D(x)}
      \d G(y)
    }{\int_{D(x)} \varphi(y) \d y}.\]
 The proof of this result is largely the same as the argument for Theorem 2 with appropriate substitutions.  Assumptions \ref{assmp:as_cont_allX} and \ref{assmp:monotonish_allX}  are needed when applying Lemma 2 to test statistic $T_n(x;D)$ to ensure that the regularity conditions needed for convergence apply to the truncated score function, and Assumption \ref{assmp:nonzero_allX} is needed to ensure $\pi(x;D)$ has a nonzero denominator.
  
\section{Additional proofs}\label{sec:prove_lemmas}

\subsection{Proof of \cref{th:worst_case}}

Throughout the proof, we condition on $\Fcal$, dropping it from the notation for
simplicity. For each $i \in [n]$, write
$S_i = \indicator{Y_{(i)} > 0}$,
$X_i = T_n\brace{i/(n+1)} = \sum_{j=n-i+1}^n c_j S_j$, and
$a_i = f_{\alpha,n}\brace{i/(n+1)}$. Under $H_0(\Gamma)$, the $(S_i)$ are
independent with $1 / (1 + \Gamma) \leq \text{pr}(S_i = 1) \leq \Gamma / (1 +
\Gamma)$. Let $p_i = \text{pr}(S_i = 1)$. We wish to show that the rejection
probability $\text{pr}(X_i \geq a_i\text{ for some }i \in [n] )$ is maximized when
$p_i = \Gamma / (1 + \Gamma)$ for all $i \in [n]$.

Write $S = (S_1, \dots, S_n)^T$, a random vector in
$\brace{0,1}^n$. For $s \in \brace{0,1}^n$,
$\text{pr}(S = s) = \prod_{i=1}^n p_i^{s_i} (1-p_i)^{1 - s_i}$. Let
$\Rcal = \ebrace{s \in \brace{0,1}^n: \sum_{j=n-i+1}^n c_j s_j \geq a_i
  \text{ for some } i \in [n]}$. This set represents the rejection event, in the
sense that the test rejects if and only if $S \in \Rcal$. We will show that
$\text{pr}(S \in \Rcal)$ is increasing in $p_i$ for each $i \in [n]$, from which it
follows that the rejection probability is maximized when $p_i$ is maximized for
each $i$.

We claim that if $s \in \Rcal$ and $s' \geq s$ elementwise, then $s' \in
\Rcal$. To see this, observe that $s \in \Rcal$ implies that we can choose
$i \in [n]$ such that $\sum_{j=n-i+1}^n c_j s_j \geq a_i$. Then
$\sum_{j=n-i+1}^n c_j s_j' \geq \sum_{j=n-i+1}^n c_j s_j \geq a_i$, so
$s' \in \Rcal$.

Now write
$\text{pr}(S \in \Rcal) = \sum_{s \in \Rcal} \prod_{i=1}^n p_i^{s_i} (1 - p_i)^{1 -
  s_i}$, and differentiate with respect to $p_k$ for any $k \in [n]$:
\begin{align}
  \frac{\d}{\d p_k} \text{pr}(S \in \Rcal)
  &= \sum_{s \in \Rcal} \ebrace{
       (2s_k - 1) \prod_{i \neq k} p_i^{s_i} (1 - p_i)^{1 - s_i}} \nonumber\\
  &= \sum_{\substack{s \in \Rcal \\ s_k = 1}} \pi^{(k)}(s)
     - \sum_{\substack{s \in \Rcal \\ s_k = 0}} \pi^{(k)}(s),
     \label{eq:two_sums}
\end{align}
where $\pi^{(k)}(s) = \prod_{i \neq k} p_i^{s_i} (1 - p_i)^{1 - s_i}$. For each
$s \in \Rcal$ with $s_k = 0$, there corresponds an $s'$ that is identical except
for $s'_k = 1$, i.e., $s'_i = s_i \indicator{i \neq k} + \indicator{i = k}$, and
this $s' \in \Rcal$ by the claim above. Also, $\pi^{(k)}(s) =
\pi^{(k)}(s')$. Hence each term in the second sum of \eqref{eq:two_sums} is
canceled by a term in the first sum. We conclude
$\frac{\d}{\d p_k} \text{pr}(S \in \Rcal) \geq 0$ as desired, completing the
proof.

We remark that an alternative proof could use Holley's inequality for
distributions over finite distributive lattices \citep[Sections 2.10,
4.7.2]{rosenbaum_observational_2002}. We have opted for the direct proof above
to keep the paper more self-contained.

\subsection{Proof of \cref{th:f_approximation}}

The following result ensures convergence of certain Riemann sums for some
unbounded functions, and is necessary to analyze the asymptotic behavior of
$f_{\alpha,n}$.

\begin{lemma}\label{th:riemann}
  Suppose $\varphi: (0, 1) \to [0, \infty)$ is discontinuous on a set of measure
  zero, $\int_0^1 \varphi(x) \d x < \infty$, and there exists a constant
  $a \in [0, 1 / 2)$ such that $\varphi$ is nonincreasing on $(0, a)$,
  nondecreasing on $(1 - a, 1)$, and bounded on $(a, 1 - a)$. Then
  \begin{align*}
    \lim_{n \to \infty} \frac{1}{n} \sum_{i=1}^n \varphi\pfrac{i}{n+1}
    = \int_0^1 \varphi(x) \d x.
  \end{align*}
\end{lemma}
\begin{proof}
  Write $\varphi = \varphi_1 + \varphi_2 + \varphi_3$ where
  $\varphi_1(x) = \varphi(x) \indicator{x < a}$,
  $\varphi_2(x) = \varphi(x) \indicator{a \leq x \leq 1 - a}$, and
  $\varphi_3(x) = \varphi(x) \indicator{x > a}$. Since $\varphi_2$ is bounded,
  it is Riemann integrable, so
  $n^{-1} \sum_{i=1}^n \varphi_2\brace{i/(n+1)} \to \int_0^1 \varphi_2(x) \d x$
  by standard Riemann integration theory, noting that
  $i/(n+1) \in ((i-1)/n, i/n)$ for each $i \in [n]$. For $\varphi_1$ and
  $\varphi_3$, we appeal to \cref{th:riemann_monotone} below to conclude that
  $n^{-1} \sum_{i=1}^n \varphi_k\brace{i/(n+1)} \to \int_0^1 \varphi_k(x) \d x$ for
  $k = 1, 3$. The result follows by linearity.
\end{proof}

\begin{lemma}\label{th:riemann_monotone}
  Suppose $\varphi: (0,1) \to [0, \infty)$ is monotone and
  $\int_0^1 \varphi(x) \d x < \infty$. Then
  \begin{align*}
    \lim_{n \to \infty} \frac{1}{n} \sum_{i=1}^n \varphi\pfrac{i}{n+1}
    = \int_0^1 \varphi(x) \d x.
  \end{align*}
\end{lemma}
\begin{proof}
  Suppose first that $\varphi$ is nondecreasing, and for each $n \in \N$ define
  $\varphi_n(x) = \varphi\brace{i/(n+1)}$ for $i/(n+1) \leq x < (i+1)/(n+1)$,
  $i = 1, \dots, n$, and $\varphi_n(x) = 0$ for $x < 1/(n+1)$. Then
  $\abs{\varphi_n} \leq \abs{\varphi}$ for all $n$ by construction, since
  $\varphi$ is nonnegative and nondecreasing. Furthermore, since $\varphi$ is
  monotone, it is discontinuous at a countable number of points
  \citep[p. 344]{knapp_basic_2007}, so $\varphi_n(x) \to \varphi(x)$ pointwise
  almost everywhere. So the dominated convergence theorem implies
  \begin{align*}
    \lim_{n \to \infty} \frac{1}{n+1} \sum_{i=1}^n \varphi\pfrac{i}{n+1}
    = \lim_{n \to \infty} \int_0^1 \varphi_n(x) \d x
    = \int_0^1 \varphi(x) \d x.
  \end{align*}
  The conclusion follows since $(n+1)/n \to 1$ as $n \to \infty$. If $\varphi$
  is instead nonincreasing, apply the above argument to
  $x \mapsto \varphi(1-x)$.
\end{proof}

Before proving \cref{th:f_approximation} we require one more technical
lemma. For $\rho \in [1/2, 1)$, define
\begin{align}
  h_\rho(x) = e^x / \ebrace{1 + \rho(e^x - 1)}^2 \quad \text{on } x \geq 0.
\end{align}
\begin{lemma}\label{th:h_technical}
  For any $\rho \in [1/2, 1)$, $0 \leq h_\rho(x) \leq 1$ for all $x \geq 0$.
\end{lemma}
\begin{proof}
  That $h_\rho(x) \geq 0$ for all $x \geq 0$ is clear from the definition. To
  see that $h_\rho(x) \leq 1$, observe
  \begin{align*}
    h_\rho'(x) = e^x
      \ebracket{\frac{1 - \rho(1 + e^x)}{\brace{1 + \rho(e^x - 1)}^3}}.
  \end{align*}
  Now the inequality $e^x \geq 1 + x$ implies
  $1 - \rho (1 + e^x) \leq 1 - 2 \rho \leq 0$ by our assumption $\rho \geq 1/2$,
  while $1 + \rho(e^x - 1) \geq 1 > 0$. Hence $h_\rho'(x) \leq 0$ for all
  $x \geq 0$. Together with $h_\rho(0) = 1$, the conclusion follows.
\end{proof}

We are now ready to prove \cref{th:f_approximation}.
\begin{proof}[of \cref{th:f_approximation}]
  The limit $n^{-1} \mu_n(x) \to \rho_\Gamma \int_{1-x}^1 \varphi(y) \d y$
  follows directly from \cref{th:riemann} applied to the function
  $q \mapsto \varphi(1-q) \indicator{q \leq x}$. The bulk of the work is in
  proving that $f_{\alpha,n}(x) = \mu_n(x) + \Ocal(\sqrt{n})$. For this, we
  start with first-order application of Taylor's theorem about $\lambda = 0$,
  which yields, for any $c \geq 0$, $\lambda \geq 0$,
\begin{align}
  \log\ebrace{1 + \rho_\Gamma(e^{c \lambda} - 1)}
  &= \rho_\Gamma c \lambda +
    \frac{\rho_\Gamma (1 - \rho_\Gamma) h_{\rho_\Gamma}(\xi) c^2 \lambda^2}{2},
      \label{eq:taylor1}
\end{align}
for some $\xi \in [0, c \lambda]$. So combining the definition
\eqref{eq:uniform_bound} of $f_{\alpha,n}$ with the expansion
\eqref{eq:taylor1}, we have
\begin{align*}
  f_{\alpha,n}(x) &= \frac{\log \alpha^{-1}}{\lambda_n} + \mu_n(x)
       + \frac{\rho_\Gamma(1-\rho_\Gamma) \lambda_n}{2}
         \sum_{i=\ceil{(1-x)(n+1)}}^n h_{\rho_\Gamma}(\xi_i) c_i^2,
\end{align*}
where $\xi_i \in [0, c_i \lambda_n]$ for each $i = 1, \dots, n$. Since
$\Gamma \geq 1$, we are assured $\rho_\Gamma \in [1/2, 1)$, so
\cref{th:h_technical} implies
\begin{align*}
  0 \leq
  \frac{\rho_\Gamma(1-\rho_\Gamma) \lambda_n}{2}
    \sum_{i=\ceil{(1-x)(n+1)}}^n h_{\rho_\Gamma}(\xi_i) c_i^2
  \leq \frac{\lambda_n \sigma_n^2(x)}{2},
\end{align*}
so that
\begin{align*}
  0 \leq f_{\alpha,n}(x) - \mu_n(x) \leq
    \frac{\log \alpha^{-1}}{\lambda_n} + \frac{\lambda_n \sigma_n^2(x)}{2}.
\end{align*}
Applying \cref{th:riemann} to the function
$q \mapsto \varphi^2(1-q) \indicator{q \leq x}$, which is integrable by
Assumption 1, we see that $n^{-1} \sigma_n^2(x) = \Ocal(1)$ for each
$x \in (0,1)$. Together with the definition of $\lambda_n$ just below
\eqref{eq:uniform_bound}, we conclude
\begin{multline}
  0 \leq \frac{f_{\alpha,n}(x) - \mu_n(x)}{\sqrt{n}}
  = \frac{1}{\sqrt{n}} \eparen{
    \frac{\log \alpha^{-1}}{\lambda_n} + \frac{\lambda_n \sigma_n^2(x)}{2}
  }
  = \sqrt{\frac{\sigma_n^2(x_0)}{2n}} +
    \sqrt{\frac{2 n \log\alpha^{-1}}{\sigma_n^2(x_0)}}
    \cdot \frac{\sigma_n^2(x)}{n}
  = \Ocal(1), \label{eq:f_approx_conclusion}
\end{multline}
as desired.
\end{proof}

\subsection{Asymptotic approximation of $f_{\alpha,n}$}
\label{sec:f_asymp_proof}

The following result formalizes the second-order expansion mentioned in
\cref{sec:intro_test}; recall the definition of $g_{\alpha,n}$ from
\eqref{eq:g_fn}.
\begin{proposition}\label{th:g_asymp}
  If $\varphi$ satisfies Assumptions
  \ref{assmp:integrable}--\ref{assmp:nonzero}, and furthermore
  $\int_0^1 \varphi^3(x) \d x < \infty$, then for all $x \in (0,1)$, we have
  $f_{\alpha,n}(x) \leq g_{\alpha,n}(x)$ for all $n$ and
  $f_{\alpha,n}(x) = g_{\alpha,n}(x) + \Ocal(1)$ as $n \to \infty$
\end{proposition}
\begin{proof}
  We follow an analogous argument to the proof of \cref{th:f_approximation},
  starting from
\begin{align*}
  \log\ebrace{1 + \rho(e^{c \lambda} - 1)}
    &= \rho c \lambda + \frac{\rho (1 - \rho) c^2 \lambda^2}{2}
       - \frac{\rho (1 - \rho) h_2(\xi) c^3 \lambda^3}{6},
\end{align*}
for some $\xi \in [0, c \lambda]$, where
\begin{align*}
  h_2(x) &= \frac{e^x \brace{\rho(1 + e^x) - 1}}{\brace{1 + \rho(e^x - 1)}^3}
\end{align*}
satisfies $0 \leq h_2(x) \leq 1$ for all $x \geq 0$. By the same argument that
led from \eqref{eq:taylor1} to \eqref{eq:f_approx_conclusion}, we find
\begin{align}
  0 \leq \frac{\log \alpha^{-1}}{\lambda_n} + \mu_n(x)
    + \frac{\lambda_n \sigma_n^2(x)}{2} - f_{\alpha,n}(x)
  \leq \frac{\rho_\Gamma(1-\rho_\Gamma) \lambda_n^2}{6}
    \sum_{i=\ceil{(1-x)(n+1)}}^n c_i^3 = \Ocal(1).
  \label{eq:f_g_relation}
\end{align}
Substituting the choice
$\lambda_n = \left\{2 \log(\alpha^{-1}) / \sigma_n^2(x_0)\right\}^{1/2}$ shows
that
\begin{align}
  \frac{\log \alpha^{-1}}{\lambda_n} + \mu_n(x)
    + \frac{\lambda_n \sigma_n^2(x)}{2}
  = \mu_n(x)
    + \ebrace{1 + \frac{\sigma_n^2(x)}{\sigma_n^2(x_0)}}
    \sqrt{\frac{\sigma_n^2(x_0) \log \alpha^{-1}}{2}} = g_{\alpha,n}(x),
  \label{eq:g_lambda}
\end{align}
so that \eqref{eq:f_g_relation} reads
$0 \leq g_{\alpha,n}(x) - f_{\alpha,n}(x) = \Ocal(1)$ as desired.
\end{proof}

The chosen value of $\lambda_n$ is the minimizer of the left-hand side
of \eqref{eq:g_lambda} when $x = x_0$, justifying the claim that $\lambda_n$ is
chosen to optimize the bound $g_{\alpha,n}(x)$ at $x = x_0$.

\subsection{Proof of \cref{th:T_convergence}}

Let $H(x) = G(x) - G(-x)$. Fix any $\epsilon > 0$. Because bounded,
continuous functions with compact support are dense in $L^p$ \citep[Theorem
13.21]{hewitt_real_1965}, we can find a continuous function
$\varphi_c: [0,1] \to [0, \infty)$ such that
$\int_0^1 \eabs{\varphi(x) - \varphi_c(x)} \d x < \epsilon$, and
$\varphi_c(x) = 0$ for all $x \in [0,b) \union (1-b,1]$ for some $0 < b <
a$. Now write
\begin{align*}
  \tau = \int_0^\infty \varphi\brace{H(x)} \d G(x), \quad 
  \tau_c = \int_0^\infty \varphi_c\brace{H(x)} \d G(x).
\end{align*}
We will show the following three relations hold:
\begin{align}
 \text{pr}\left\{ \limsup_{n \to \infty} \frac{1}{n} \eabs{
      \sum_{i=1}^n \varphi\pfrac{i}{n+1} \indicator{Y_{(i)} > 0}
      - \sum_{i=1}^n \varphi_c\pfrac{i}{n+1} \indicator{Y_{(i)} > 0}}
    < \epsilon\right\} = 1,& \label{eq:T_conv_step1}\\
 \text{pr}\left\{ \lim_{n \to \infty} \frac{1}{n} \sum_{i=1}^n \varphi_c\pfrac{i}{n+1} \indicator{Y_{(i)} > 0}
  = \tau_c\right\} = 1,
    \label{eq:T_conv_step2} \\
  \abs{\tau_c - \tau} < \epsilon.& \label{eq:T_conv_step3}
\end{align}
From this we conclude
$\limsup_{n \to \infty} \eabs{n^{-1} \sum_{i=1}^n \varphi\pfrac{i}{n+1}
  \indicator{Y_{(i)} > 0} - \tau} < 2\epsilon$  with probability one.  Since $\epsilon$ was
arbitrary, the conclusion follows.

To obtain \eqref{eq:T_conv_step1}, use the triangle inequality to write
\begin{align*}
  \frac{1}{n} \eabs{
      \sum_{i=1}^n
      \ebrace{\varphi\pfrac{i}{n+1} - \sum_{i=1}^n \varphi_c\pfrac{i}{n+1}}
      \indicator{Y_{(i)} > 0}}
  &\leq \frac{1}{n} \sum_{i=1}^n \eabs{\varphi - \varphi_c}\pfrac{i}{n+1} \\
  &= \frac{1}{n}
       \sum_{i=1}^{n+1} \eabs{\varphi - \varphi_c}\pfrac{i}{n+1}
     - \frac{\eabs{\varphi - \varphi_c}(1)}{n} \\
  &\to \int_0^1 \eabs{\varphi - \varphi_c}(x) \d x < \epsilon,
\end{align*}
where the limit uses \cref{th:riemann}, noting that $\eabs{\varphi - \varphi_c}$
is bounded on $[b,1-b]$ and monotone elsewhere, and final inequality follows
from our choice of $\varphi_c$.

The second step \eqref{eq:T_conv_step2} follows from Theorem 1 of
\citet{sen_convergence_1970} applied to $\varphi_c$, which we partially
restate. See \cref{sec:remark_sen} below for an explanation of why our statement
differs from Sen's.
\begin{lemma}[\citealp{sen_convergence_1970}, Theorem 1]\label{th:sen}
  Suppose $\varphi_c \in L^1(0,1)$ is bounded and continuous, and suppose
  $Y_1, Y_2, \dots$ are drawn i.i.d.\ from a continuous distribution $G$. Then
  \begin{align*}
  \frac{1}{n} \sum_{i=1}^n \varphi_c\pfrac{i}{n+1} \indicator{Y_{(i)} \geq 0}
  \convas \int_0^\infty \varphi_c\brace{H(x)} \d G(x).
  \end{align*}
\end{lemma}

Finally, to see \eqref{eq:T_conv_step3}, use the triangle inequality to write
\begin{align*}
  \eabs{\tau_c - \tau}
    &\leq \int_0^\infty \eabs{\varphi_c - \varphi}\brace{H(y)} \d G(y) \\
    &\leq \int_0^\infty \eabs{\varphi_c - \varphi}\brace{H(y)} \d H(y),
\end{align*}
since $H'(y) = G'(y) + G'(-y) \geq G'(y)$ and the integrand is nonnegative. From
this we conclude
\begin{align*}
  \eabs{\tau_c - \tau} \leq \int_0^1 \eabs{\varphi_c - \varphi}(u) \d u
    < \epsilon,
\end{align*}
by our choice of $\varphi_c$. The proof is complete.

\subsection{Proof of \cref{th:limiting_ds}}

  Write $q_x$ for the $x$-quantile of $\abs{Y}$ when $Y \sim G$, so that $q_x$
  is defined by the equation $G(q_x) - G(-q_x) = x$. We shall require the
  derivative of $q_x$ below, which we find by implicit differentiation:
  \begin{align}
    \frac{\d q_x}{\d x} = \frac{1}{g(q_x) + g(-q_x)}.
    \label{eq:q_derivative}
  \end{align}
  Now observe that, using the
  definition of $q_x$, we may write $\pi(x)$ from Theorem 2 as
  \begin{align}
    \pi(x) = \frac{\int_{q_{1-x}}^\infty \varphi\brace{G(y) - G(-y)} \d G(y)}
                  {\int_{1-x}^1 \varphi(y) \d y}. \label{eq:pi_with_q}
  \end{align}
  We apply the generalized form of L'H\^opital's rule, which says that
  $\limsup f/g \geq \liminf f'/g'$ when $\lim f = \lim g = 0$, to the formula
  \eqref{eq:pi_with_q} for $\pi(x)$ to find
  \begin{align*}
    \limsup_{x \downarrow 0} \pi(x)
    &\geq \liminf_{x \downarrow 0}
       \frac{\frac{\d}{\d x}\int_{q_{1-x}}^\infty \varphi\brace{G(y) - G(-y)} \d G(y)}
            {\frac{\d}{\d x} \int_{1-x}^1 \varphi(y) \d y} \\
    &= \liminf_{x \downarrow 0}
      \frac{\varphi(1-x) g(q_{1-x})}{\varphi(1-x)}
      \cdot \frac{1}{g(q_{1-x}) + g(-q_{1-x})},
  \end{align*}
  where the equality uses the fundamental theorem of calculus and
  \eqref{eq:q_derivative}. Assumption 4 on $\varphi$ implies $\varphi(q)$ must
  be positive on a neighborhood $q \in (1-\epsilon, 1)$ for some $\epsilon > 0$,
  which ensures the limit is well-defined. Reparametrizing in terms of
  $q = q_{1-x}$, and noting that $q_{1-x} \uparrow \infty$ as $x \downarrow 0$
  since $g$ is positive throughout $\R$, we have
  \begin{align*}
    \limsup_{x \downarrow 0} \pi(x)
    &\geq \liminf_{q \uparrow \infty}
      \frac{1}{1 + \frac{g(-q)}{g(q)}}.
  \end{align*}
  Hence
  $\limsup_{x \downarrow 0} \frac{\pi(x)}{1 - \pi(x)} \geq \liminf_{q \uparrow
    \infty} \frac{g(q)}{g(-q)}$. The conclusion follows from
  \cref{th:unif_ds}.

\subsection{Proof of \cref{th:normal_scores_lp}}

Fix any $p \geq 1$. A standard Cram\'er-Chernoff tail bound for the normal
distribution \citep[Section 2.2]{boucheron_concentration_2013} gives
$1 - \Phi(x) \leq e^{-x^2/2}$, which implies
$\Phi^{-1}(q) \leq \sqrt{2 \log (1-q)^{-1}}$. Hence
\begin{align*}
  \int_0^1 \eabs{\varphi(q)}^p \d q
    &\leq 2^{p/2} \int_0^1 \ebracket{\log\brace{2/(1-q)}}^{p/2} \d q \\
    &= 2^{1+p/2} \int_{\log 2}^\infty y^{p/2} e^{-y} \d y
\end{align*}
using the substitution $y = \log\brace{2/(1-q)}$. The final integral is upper
bounded by $\Gamma(1+p/2)$, using the definition of the Gamma function and
non-negativity of the integrand, which completes the proof.

\subsection{A remark on Theorem 1 from \citet{sen_convergence_1970}}
\label{sec:remark_sen}

\Citet{sen_convergence_1970} assumes only that $\varphi \in L^1(0,1)$ is
continuous. Denoting $\varphi_n(x) = \varphi\brace{i/(n+1)}$ for
$(i-1)/n < x \leq i/n$, $i = 1, \dots, n$, their proof (p. 2141) claims that
\begin{align}
  \lim_{n \to \infty} \int_0^1 \abs{\varphi_n(x)} \d x
  = \int_0^1 \abs{\varphi(x)} \d x. \label{eq:sen_convergence}
\end{align}
The conclusion \eqref{eq:sen_convergence} is not true for all continuous
$\varphi \in L^1(0,1)$, as the counterexample below shows. However, noting that
$\int_0^1 \varphi_n(x) \d x = n^{-1} \sum_{i=1}^n \varphi\brace{i/(n+1)}$, our
\cref{th:riemann} shows that \eqref{eq:sen_convergence} is true under stronger
conditions, and in particular is true for bounded $\varphi$. This is the reason
we require boundedness in our restatement of Sen's result, \cref{th:sen}.

Let $\varphi(x) = n$ for $1/(n+1) \leq x \leq 1/(n+1) + 1/(n2^{n+1})$,
$n \in \N$. Then $\int_0^1 \varphi(x) \d x = \sum_{n=1}^\infty 2^{-n-1} = 1/2$,
hence $\varphi(x) \in L^1$. But
$n^{-1} \sum_{i=1}^n \varphi\brace{i/(n+1)} \geq n^{-1} \varphi\brace{1/(n+1)} =
1$ for all $n$, so
$\liminf_{n \to \infty} n^{-1} \sum_{i=1}^n \varphi\brace{i/(n+1)} \geq 1 > 1/2
= \int_0^1 \varphi(x) \d x$, showing \eqref{eq:sen_convergence} does not
hold. This $\varphi$ is not continuous, but may be replaced with a continuous
approximation by standard arguments.

\section{Simulations of Type I error rate}\label{sec:typeI}

\begin{figure}
  \includegraphics{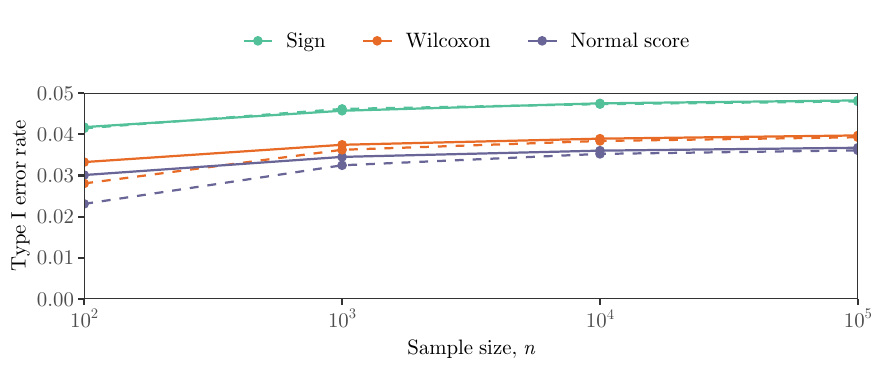}
  \caption{Simulated Type I error rate of uniform sign rank tests versus sample
    size, based on one million replications. Solid lines show $\Gamma = 1$ while
    dashed lines show $\Gamma = 5$. All tests use $x_0 = 1/3$ and nominal level
    $\alpha = 0.05$. \label{fig:typeI_size}}
\end{figure}

\begin{figure}
  \includegraphics{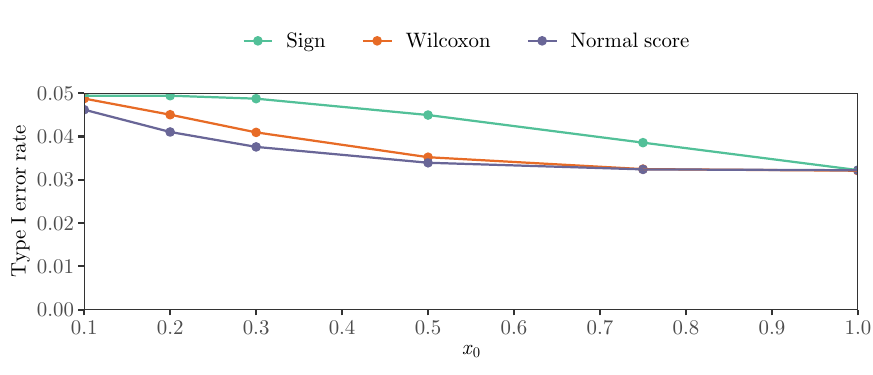}
  \caption{Simulated Type I error rate of uniform sign rank tests versus tuning
    parameter $x_0$ at sample size $n = 10^5$ when $\Gamma = 1$, based on one
    million replications. All tests use nominal level $\alpha =
    0.05$. \label{fig:typeI_x0}}
\end{figure}

\Cref{th:uniform_bound} shows that the uniform general signed rank test has size
no larger than the nominal level $\alpha$, while the simulations in
\cref{sec:simulations} illustrate the realized power of the test under various
alternatives. Neither tells us how close the true size of the test is to the
nominal level $\alpha$, that is, how close the inequality in
\cref{th:uniform_bound} is to an equality. Figures~\ref{fig:typeI_size} and
\ref{fig:typeI_x0} present simulation results under the null $H_0(\Gamma)$ to
illustrate the size of some uniform tests. These simulations use the worst-case
distribution within $H_0(\Gamma)$ defined by \cref{th:worst_case}. 

The simulations reveal several interesting patterns. Size appears to increase towards a limit as sample size grows, but this
  limit is not equal to $\alpha$ in general.
For larger $\Gamma$, the test becomes more conservative at small sample
  sizes, but the effect is negligible at large sample sizes.
As $x_0$ approaches zero, the asymptotic size approaches $\alpha$.
Finally, the sign score function yields the least conservative uniform test,
  followed by the Wilcoxon score function and then the normal score function.

\section{Additional plots and tables}\label{sec:addplots}

\begin{table}[ht]
\centering
\small
\begin{tabular}{r|rr|r}
\hline
& \multicolumn{2}{c|}{Average attribute values} & Standardized \\
Variable & 15+ fish servings / mo & 0-2 fish servings / mo & difference \\
\hline
Age & 43.73 & 43.63 & 0.005 \\
Household Income/(2x poverty line) & 2.99 & 2.96 & 0.017 \\
Female & 0.46 & 0.46 & 0.004 \\
Hispanic & 0.19 & 0.18 & 0.002 \\
Black & 0.22 & 0.22 & 0.001 \\
Smoker & 0.44 & 0.42 & 0.011 \\
Cigarettes/Day & 4.09 & 4.04 & 0.011 \\
High School Graduate & 0.80 & 0.80 & 0.000 \\
Missing HS Graduation Status & 0.03 & 0.03 & 0.000 \\
\hline
\end{tabular}
\caption{Balance table for 1,672 matched pairs formed from NHANES data.  Each
  pair contains one individual who consumed $\geq$15 servings of fish in the
  previous month, and one who consumed no more than two.  The first two columns
  give the sample means in the matched samples for various attributes of
  interest, and the third gives the standardized difference, which is computed
  by dividing the difference in group sample means by the pooled standard
  deviation estimate from the full dataset before matching. \label{tab:balance}}
\end{table}

\Cref{tab:balance} shows covariate balance between treated and control groups
for the constructed pairs used in the data example of Section 6.
\begin{figure}
  \includegraphics{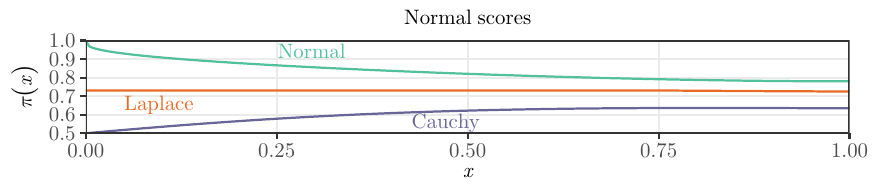}
  \caption{$\pi(x)$ from Theorem 2 for the normal score function when $G$ is
    standard normal, Laplace (double exponential) or Cauchy, and $\tau =
    1/2$. \label{fig:pi_plot_extra}}
\end{figure}
\Cref{fig:pi_plot_extra} plots $\pi(x)$ as defined in Theorem 2 for
additional score functions not included in Figure 2. See
Section 4 for discussion.

\begin{figure}
\includegraphics{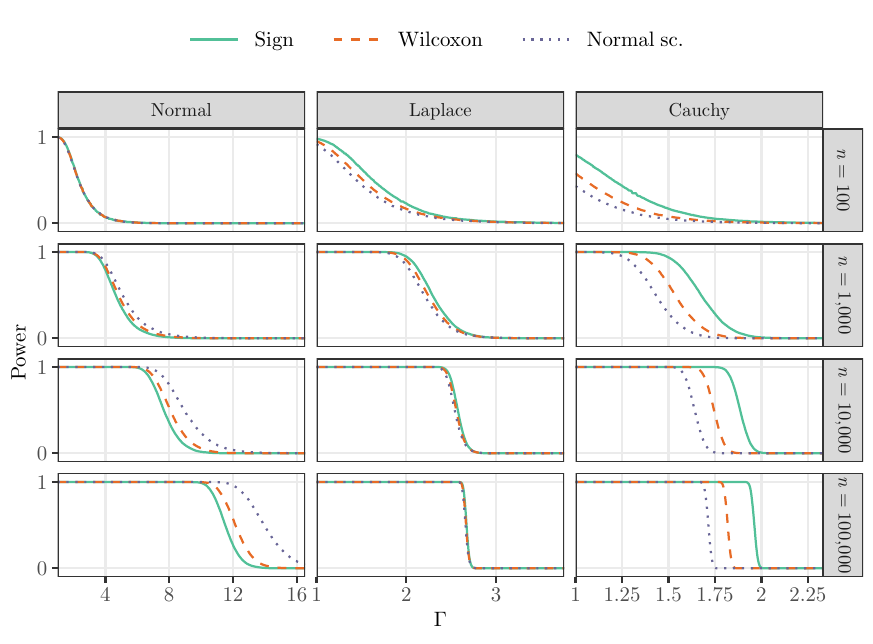}
\caption{Comparison of simulated power for uniform tests using different score
  functions, based on 10,000 replications under alternative model $H_1(G)$ with
  $G$ as indicated, having center $1/2$ and unit scale. All tests use
  $\alpha = 0.05$. \label{fig:compare_scores}}
\end{figure}

\Cref{fig:compare_scores} complements Figure 3 by comparing
power between uniform tests with different score functions. Tests tend to
perform similarly with small sample sizes, but clear distinctions emerge with
large sample sizes. In the normal case, the normal scores test dominates while
the redescending score function substantially underperforms. As we have seen,
under normal noise the outliers contain the most information, and a score
function that places more weight upon pairs with large absolute differences
will attain higher power as a result. Conversely, in the Cauchy case, the normal
scores tests performs the worst, while the sign test performs the best. Here the
extreme tails yield less information, as indicated by Figure 2. The
Laplace case is a middle ground in which the tails yield no more or less
information than most of the rest of the distribution, as we have seen in
Figure 2. Here the choice of score function makes little difference.

\end{document}